\documentclass[journal]{IEEEtran}
\usepackage[cmex10]{amsmath}
\usepackage{amssymb,amsfonts,amsthm}
\usepackage{booktabs}
\usepackage{graphicx}
\usepackage{cite}
\usepackage[colorlinks=false,pdfborder={0 0 0}]{hyperref}
\newcommand{\customfootnotetext}[2]{{
  \renewcommand{\thefootnote}{#1}
  \footnotetext[0]{#2}}}
\newtheorem{theorem}{Theorem}
\newtheorem{lemma}[theorem]{Lemma}
\newtheorem{proposition}[theorem]{Proposition}
\newtheorem{corollary}[theorem]{Corollary}
\newtheorem{definition}[theorem]{Definition}
\newtheorem{remark}[theorem]{Remark}
\newtheorem{fact}[theorem]{Fact}
\newcommand{\ket}[1]{\lvert #1\rangle}
\newcommand{\bra}[1]{\langle #1\rvert}
\newcommand{\Tr}{\operatorname{Tr}}
\DeclareMathOperator{\Perm}{Perm}
\newcommand{\bbE}{\mathbb E}
\newcommand{\QED}{\ensuremath{\square}}
\begin{document}
\title{Exact Posterior Prediction from Product Haar Measurements and a Randomized-Mesh Maximum-Likelihood Bridge}
\author{ Naqueeb Ahmad Warsi\textsuperscript{$*$}, Ayanava Dasgupta\textsuperscript{$*$}
and Masahito Hayashi\textsuperscript{$\dagger$} 
}
\date{}              

\maketitle

\customfootnotetext{$*$}{
Indian Statistical Institute,
Kolkata 700108, India.
(e-mail: [naqueebwarsi, ayanavadasgupta\_r]@isical.ac.in
)}
\customfootnotetext{$\dagger$}{
School of Data Science, The Chinese University of Hong Kong, Shenzhen, Longgang District, Shenzhen, 518172, China,
International Quantum Academy (IQA), Futian District, Shenzhen 518048, China,
and
Graduate School of Mathematics, Nagoya University, Chikusa-ku, Nagoya 464-8602, Japan.
(e-mail: hmasahito@cuhk.edu.cn).
}

\begin{abstract}
We study prediction of one unmeasured copy of an unknown finite-dimensional pure quantum state after independently measuring the observed copies with the one-copy Haar POVM. For the resulting fixed separable observation, the Bayes predictive state under quantum relative-entropy loss is the full-rank posterior mean. We evaluate it exactly through permanents of minors of the outcome Gram matrix, prove positivity and normalization of the permanent representation, and establish minimaxity among decision rules based on these fixed outcomes. Using Hilbert--Schmidt projection of the posterior mean and an exactly analyzed linear-inversion competitor, we prove a finite-sample posterior-purity bound and an explicit relative-entropy guarantee with sample count $O((d^2/\epsilon)\log(d/\epsilon))$. Using the exact collective benchmarks proved in a companion paper, we also obtain entropy and posterior-purity comparisons. To connect the fine product data to a regular ray-valued experiment without assuming deterministic finite-sample uniqueness, we draw one random orientation for a finite projective mesh, retain it for the entire sample, and apply a selected finite-alphabet MLE in the resulting reference experiment. A nested recovering mesh uniformly recovers the fine Fisher information. Consequently every sufficiently fine fixed mesh gives a regular finite-alphabet model. Its selected MLE has Haar-averaged inverse-Fisher coefficient $\overline a_k$, and refining the mesh after the fixed-mesh large-sample limit makes $\overline a_k$ approach $d-1$. This yields epsilon-optimal leading relative-entropy and posterior-overlap bounds for the fine posterior and the sharper fixed-dimensional asymptotic scale $(d/\epsilon)\log(d/\epsilon)$. The latter is not asserted as a uniform finite-sample guarantee. No deterministic finite-sample MLE uniqueness or triangular mesh schedule is used.
\end{abstract}

\section{Introduction}
In quantum inference, measurements on several identically prepared systems may
be used not only to estimate their common state, but also to assign a state to
another copy that remains unmeasured
\cite{SchackBrunCaves2001QuantumBayes,TanakaKomaki2005BayesianPredictive}.
The latter is a prediction problem rather than merely a point-estimation
problem.  Its decision is a density operator conditioned on the observed data,
not necessarily a pure-state estimate, and it can therefore represent both the
information extracted from the measured copies and the uncertainty that
remains.  We study this problem for an unknown finite-dimensional pure state
with Haar prior.  The first $n$ copies are measured independently by the
one-copy Haar POVM, one additional copy is left unmeasured, and the announced
state for that future copy is evaluated by quantum relative-entropy loss.

The general Bayesian foundation is known.  In classical predictive-density
estimation, the posterior predictive density is the Bayes act under
Kullback--Leibler loss; see, for example,
\cite{Aitchison1975PredictionFit,AitchisonDunsmore1975Prediction,Geisser1993Predictive,Komaki2006ShrinkagePrediction,GeorgeLiangXu2006PredictiveKL}.
In the quantum setting, Tanaka and Komaki formulated Bayesian predictive
density operators and established their optimality under averaged quantum
relative entropy, while the quantum Bayes rule explains how posterior states
arise after measurements
\cite{TanakaKomaki2005BayesianPredictive,SchackBrunCaves2001QuantumBayes}.
Accordingly, we do not claim as new the fixed-observation principle that the
posterior mean is the Bayes predictive state.  Our problem is to evaluate that
state exactly for the product Haar observation and to determine its
finite-sample and asymptotic performance.

A second line of prior work concerns pure-state estimation under fidelity-based
criteria.  Optimal extraction of a pure-state estimate from a finite ensemble
and realizable optimal covariant measurements were developed by Massar and
Popescu and by Derka, Bu\v{z}ek, and Ekert
\cite{MassarPopescu1995OptimalEstimation,DerkaBuzekEkert1998OptimalEstimation}.
Hayashi subsequently developed an asymptotic estimation theory for
finite-dimensional pure-state models and identified the relevant
information-geometric structure
\cite{Hayashi1998PureStateEstimation}.  These estimation results concern a
reported ray and its fidelity or local quadratic loss.  Our decision object is
instead a generally mixed posterior predictive state, and our primary loss is
quantum relative entropy.  Nevertheless, pure-state estimation enters our
argument through a covariant ray-valued coarse observation: its mean
infidelity determines the spectrum of the coarse posterior and thereby bounds
the entropy and purity of the posterior conditioned on the complete product
Haar data.

The relation to \cite{Hayashi1998PureStateEstimation} is particularly direct at
the level of the leading coefficient.  Under our normalization, the classical
Fisher information of the one-copy continuous Haar observation is
$2g_{\rm FS}$, whereas pure-state infidelity has local quadratic form
$g_{\rm FS}$.  Thus an efficient ray estimator has information-geometric
coefficient
\begin{equation}
 \operatorname{Tr}\!\left[g_{\rm FS}(2g_{\rm FS})^{-1}\right]
 =\frac12\dim_{\mathbb R}\mathbb{CP}^{d-1}=d-1.
 \label{eq:intro_hayashi_coefficient}
\end{equation}
Our purpose is not merely to repeat this estimation coefficient.  We construct
a regular finite-alphabet approximation that transfers it to expected
infidelity without applying ordinary smooth MLE theory directly to the
continuous density, which vanishes on a parameter-dependent projective
hyperplane.  Covariant posterior coarse-graining then transfers the same
coefficient to posterior entropy and posterior purity for the complete product
Haar data.

The unrestricted collective prediction problem is treated separately in a
companion paper \cite{CompanionGlobal}.  That work proves the exact global
relative-entropy endpoint
\[
 h_d\!\left(\frac{n+1}{n+d}\right)
\]
and the exact global posterior-overlap endpoint
\[
 \frac{(n+1)^2+d-1}{(n+d)^2}.
\]
We use these values only as certified benchmarks.  We do not repeat the
arbitrary-effect posterior identity, pointwise majorization, global Bayes
proof, global minimax proof, or unrestricted posterior-overlap optimization.
The distinction between the two problems is essential: the companion paper
optimizes over collective measurements, whereas the present paper fixes the
independent one-copy Haar measurement and optimizes and analyzes the decision
rule based on its outcomes.

Our first contribution is the exact posterior predictive state for the product
Haar data.  Symmetric projection converts the posterior integral into
Gram-matrix combinatorics: the outcome-dependent one-particle marginal is
expressed through permanents of Gram minors, and the final posterior state adds
a universal identity term.  The result is a full-rank state with spectral floor
$I/(n+d)$, not a naive average of the observed projectors.  We prove positivity
and normalization of the permanent representation and show that the posterior
state is Bayes optimal and minimax among decision rules based on this fixed
product observation.  This local minimax statement does not optimize over
collective measurements.

Our second contribution is a uniform finite-sample performance guarantee.
Writing $\mathcal P_{n,d}^{\rm loc}$ for the mean posterior purity and
$R_{n,d}^{\rm loc}$ for the relative-entropy risk, we prove
\begin{equation}
 1-\mathcal P_{n,d}^{\rm loc}
 \le \frac{d^2+d-2}{n}
 \label{eq:intro_finite_purity_bound}
\end{equation}
and
\begin{equation}
 R_{n,d}^{\rm loc}
 \le \frac{d^2+d-2}{n}\frac{n+d}{n+d-1}\log(n+d).
 \label{eq:intro_finite_entropy_bound}
\end{equation}
Consequently, $R_{n,d}^{\rm loc}\le\epsilon$ is guaranteed by the explicit
sufficient condition
\begin{equation}
 n\ge \frac{4(d^2+d-2)}{\epsilon}
 \log\!\left(\frac{4(d^2+d-2)(d+1)}{\epsilon}\right),
 \label{eq:intro_sample_complexity}
\end{equation}
which gives the uniform finite-sample scaling
$O((d^2/\epsilon)\log(d/\epsilon))$.  This result follows from the
Hilbert--Schmidt projection property of the posterior mean and the exactly
computed variance of a linear-inversion competitor; it is independent of the
randomized-mesh asymptotics.

Our third contribution is an asymptotically sharp evaluation at every fixed
dimension.  For a fixed finite projective mesh, one Haar-random orientation is
drawn and retained for the entire sample.  Conditional on that orientation,
the cell labels form a regular finite-alphabet curved exponential family.  A
fixed measurable MLE selector is asymptotically efficient in this reference
model, while multiplication by the shared orientation makes the reported-ray
experiment exactly covariant.  The fixed-mesh expected-infidelity coefficient
is $\overline a_k$, and mesh refinement recovers the Hayashi coefficient through
$\overline a_k\to d-1$.  Combining the resulting coarse-posterior upper bounds
with the exact collective benchmarks gives
\begin{align}
 R_{n,d}^{\rm loc}
 &=(d-1)\frac{\log n}{n}
   +o\!\left(\frac{\log n}{n}\right),
 \label{eq:intro_sharp_entropy_asymptotic}\\
 \mathcal P_{n,d}^{\rm loc}
 &=1-\frac{2(d-1)}{n}+o(n^{-1}).
 \label{eq:intro_sharp_purity_asymptotic}
\end{align}
Thus the local posterior has the same leading entropy and purity coefficients
as the exact collective benchmarks, although the measurements and operational
optimizations are different.

The two performance conclusions have different scopes.  The sample-count bound
\eqref{eq:intro_sample_complexity} is a uniform finite-sample guarantee with
quadratic dimensional dependence.  The sharper formulas
\eqref{eq:intro_sharp_entropy_asymptotic}--\eqref{eq:intro_sharp_purity_asymptotic}
are fixed-dimensional asymptotic statements.  Their order of limits is
essential: one first fixes the mesh and lets $n\to\infty$, and only afterward
refines the mesh.  We make no claim for a triangular schedule $k=k_n$ or for a
uniform nonasymptotic linear-in-dimension theorem.

The paper is organized as follows. Section~\ref{sec:bayes_tools_p2} gives the
fixed-observation posterior and coarse-graining tools.
Section~\ref{sec:local_haar_povm_prediction} derives the exact posterior, its
permanent representation, and the decision-rule minimax statement.
Section~\ref{sec:finite_sample_risk_complexity} proves the finite-sample risk and
sample-complexity bounds. Sections~\ref{sec:rmesh_main}--\ref{sec:fixed_mesh_mle_efficiency}
introduce recovering meshes and derive fixed-mesh MLE efficiency from their
regularity. Section~\ref{sec:posterior_comparison_asymptotics} recovers the fine
coefficient, derives the covariant posterior comparisons and sharp
iterated-limit asymptotics, and concludes with a finite-sample numerical
illustration.
Appendix~\ref{app:informationally_complete_meshes} constructs recovering meshes and verifies uniform Fisher recovery.  The remaining appendices give the permanent calculations.
\section{Posterior Prediction for a Fixed Observation}
\label{sec:bayes_tools_p2}
Let $M^{(n)}(dx)$ be a fixed POVM on the observed copies, let $\pi_M(d\phi\mid x)$ be the posterior under the Haar prior, and define
\[
 \mu_M(x)=\int \rho_\phi\,\pi_M(d\phi\mid x).
\]
For output-marginal almost every $x$ and every density matrix $\sigma$ whose support contains that of $\mu_M(x)$,
\[
 \int D(\rho_\phi\Vert\sigma)\,\pi_M(d\phi\mid x)
 =S(\mu_M(x))+D(\mu_M(x)\Vert\sigma).
\]
Thus $\mu_M(x)$ is the Bayes act and the optimized conditional loss is $S(\mu_M(x))$. This fixed-observation statement does not optimize the measurement.

Suppose a fine observation $X$ is randomized to a coarse observation $Y$ by a Markov kernel, and define regular conditional posterior states
\[
 \rho_X=\bbE[\rho_\Phi\mid X],\qquad
 \rho_Y=\bbE[\rho_\Phi\mid Y]=\bbE[\rho_X\mid Y].
\]
Then
\begin{align}
 \bbE S(\rho_Y)-\bbE S(\rho_X)
 &=\bbE D(\rho_X\Vert\rho_Y),\label{eq:p2_entropy_gap}\\
 \bbE\Tr(\rho_X^2)-\bbE\Tr(\rho_Y^2)
 &=\bbE\Tr[(\rho_X-\rho_Y)^2].\label{eq:p2_purity_gap}
\end{align}
Hence coarse-graining increases posterior entropy and decreases posterior purity.

Define
\[
 h_d(r)=-r\log r-(1-r)\log\frac{1-r}{d-1},\quad
 p_d(r)=r^2+\frac{(1-r)^2}{d-1}.
\]
For an exactly covariant ray-valued coarse observation $\widehat\psi$, stabilizer invariance gives, output-marginal almost everywhere,
\[
 \rho(\widehat\psi)=r\rho_{\widehat\psi}+\frac{1-r}{d-1}(I-\rho_{\widehat\psi}).
\]
Its entropy is $h_d(r)$ and its purity is $p_d(r)$.

\section{Posterior Predictive State for the Local Haar-POVM Model}
\label{sec:local_haar_povm_prediction}

This section starts from the outcomes produced by independent applications
of the one-copy Haar POVM and constructs the corresponding posterior
predictive-state rule.  The product measurement is specified before the
decision rule is optimized; no optimization over measurements is performed
in this section.

\subsection{The one-copy and product Haar-POVM setting}
\label{subsec:local_product_haar_povm_model}

Throughout this section, \(du\) denotes the normalized unitarily invariant
probability measure on the complex projective space of pure-state rays.  We
identify a ray \(u\) with the rank-one projector \(\ket{u}\bra{u}\).  With this
normalization,
\begin{equation}
    \int \ket{u}\bra{u}\,du
    =
    \frac{I}{d}.
    \label{eq:haar_projector_first_moment_main}
\end{equation}
The one-copy Haar POVM is
\begin{equation}
    E(du)
    =
    d\,\ket{u}\bra{u}\,du,
    \label{eq:haar_povm_normalization_main}
\end{equation}
so that \(\int E(du)=I\).  For background on Haar-measure tools and
group-theoretic methods in quantum information, see
\cite{Mele2024,Group2,Wilde_2013,H2017QIT}.

The local observation model used in this section is the product POVM obtained by
applying \eqref{eq:haar_povm_normalization_main} independently to the \(n\)
input copies:
\begin{equation}
    E^{(n)}(d\vec u)
    :=
    E(du_1)\otimes\cdots\otimes E(du_n),
    ~
    \vec u=(u_1,\ldots,u_n).
    \label{eq:product_haar_povm_def}
\end{equation}
For each outcome sequence, set
\begin{equation}
    \Pi_i:=\ket{u_i}\bra{u_i},
    \quad
    O_n(\vec u):=\Pi_1\otimes\cdots\otimes\Pi_n.
    \label{eq:Pi_i_On_def}
\end{equation}
Then
\begin{equation}
    E^{(n)}(d\vec u)
    =
    d^n O_n(\vec u)\,du_1\cdots du_n.
    \label{eq:product_haar_density_kernel}
\end{equation}
Thus \(O_n(\vec u)\) is the rank-one tensor kernel of the product Haar-POVM
density, with the scalar factor \(d^n\) separated out.

For a hypothesis pure state \(\rho_\phi=\ket\phi\bra\phi\), the likelihood
density with respect to the product Haar measure is
\begin{equation}
\begin{aligned}
    p_\phi^{(n)}(\vec u)
    &:=
    \frac{
        \operatorname{Tr}
        [\rho_\phi^{\otimes n}E^{(n)}(d\vec u)]
    }{du_1\cdots du_n} \\
    &=
    d^n\operatorname{Tr}
    [\rho_\phi^{\otimes n}O_n(\vec u)]                        
    =
    d^n\prod_{i=1}^n |\langle u_i\mid\phi\rangle|^2.
\end{aligned}
    \label{eq:local_product_haar_likelihood_density}
\end{equation}

\begin{lemma}[Covariance of the product Haar-POVM]
\label{lem:product_haar_povm_covariance}
For every \(W\in U(d)\),
\begin{equation}
    E^{(n)}(d(W\vec u))
    =
    W^{\otimes n}E^{(n)}(d\vec u)W^{\otimes n\dagger},
    \label{eq:product_haar_povm_covariance}
\end{equation}
where \(W\vec u=(Wu_1,\ldots,Wu_n)\).  Equivalently,
\begin{equation}
    p_{W\phi}^{(n)}(W\vec u)
    =
    p_\phi^{(n)}(\vec u)
    \label{eq:product_haar_likelihood_covariance}
\end{equation}
with respect to the product Haar measure.
\end{lemma}

\begin{proof}
For one copy, Haar invariance gives \(d(Wu)=du\), and
\(\ket{Wu}\bra{Wu}=W\ket u\bra u W^\dagger\).  Hence
\begin{equation*}
    E(d(Wu))
    =
    W E(du) W^\dagger.
\end{equation*}
Taking tensor products over the \(n\) independent copies gives
\eqref{eq:product_haar_povm_covariance}.  The likelihood covariance
\eqref{eq:product_haar_likelihood_covariance} follows immediately from
\eqref{eq:local_product_haar_likelihood_density} and the unitary invariance of
inner products.
\end{proof}

\subsection{Posterior predictive state for product Haar-POVM data}
\label{subsec:posterior_mean_fixed_product_haar_povm}
To proceed with Bayesian inference, we use the unitarily invariant Haar measure
\(d\phi\) on the complex projective space of pure-state rays.  The general
Bayesian predictive-state principle already identifies the posterior mean as
the Bayes act.  The model-specific contribution of this section is instead the
exact evaluation of that posterior mean for the fixed product Haar-POVM and,
equivalently, the explicit evaluation of the normalized one-particle marginal
of its compressed outcome effect.

The Bayesian posterior predictive state $\hat{\rho}_n(\vec{u})$ is the exact, continuous integral of the hypothesis state $|\phi\rangle\langle\phi|$ over the hypothesis space, weighted by the normalized posterior distribution. By Bayes' theorem,
\begin{align*}
\hat{\rho}_n(\vec{u}) &= \frac{ \int |\phi\rangle\langle\phi| \Tr\left[(|\phi\rangle\langle\phi|)^{\otimes n} O_n\right] d\phi }{ \int \Tr\left[(|\phi\rangle\langle\phi|)^{\otimes n} O_n\right] d\phi } 
\\
&= \frac{ \Tr_{1:n}\left[ (O_n \otimes I_{n+1}) \int (|\phi\rangle\langle\phi|)^{\otimes n+1} d\phi \right] }{ \Tr\left[ O_n \int (|\phi\rangle\langle\phi|)^{\otimes n} d\phi \right] }.
\end{align*}

\subsection{Evaluating the denominator and Haar integrals}
The required Haar integral is the standard symmetric-subspace identity
associated with Schur--Weyl duality; see, e.g.,
\cite{Group2,Hayashi2017GroupRepresentationQuantumTheory}.  For every \(t\),
the Haar average of \((|\phi\rangle\langle\phi|)^{\otimes t}\) is proportional
to the projector onto the fully symmetric subspace.

The normalization is the dimension of that subspace, giving
\begin{equation*}
\int (|\phi\rangle\langle\phi|)^{\otimes t} d\phi = \frac{\Pi_{\mathrm{sym}}^{(t)}}{\binom{t+d-1}{t}}.
\end{equation*}

Substituting this exact identity into both the numerator ($t = n+1$) and the denominator ($t = n$) of our estimator yields,
\begin{align*}
\hat{\rho}_n(\vec{u}) =& \frac{\binom{n+d-1}{n}}{\binom{n+d}{n+1}} \frac{ \Tr_{1:n}\left[ (O_n \otimes I_{n+1}) \Pi_{\mathrm{sym}}^{(n+1)} \right] }{ \Tr\left[ O_n \Pi_{\mathrm{sym}}^{(n)} \right] } \\
=& \frac{n+1}{n+d} \frac{ \Tr_{1:n}\left[ (O_n \otimes I_{n+1}) \Pi_{\mathrm{sym}}^{(n+1)} \right] }{ \Tr\left[ O_n \Pi_{\mathrm{sym}}^{(n)} \right] }.
\end{align*}

Thus the posterior integrals reduce to finite-dimensional symmetric-subspace calculations.

\subsection{Evaluating the numerator via SWAP decomposition}
To evaluate the trace in the numerator, we must take the partial trace of the $(n+1)$-system symmetric projector over the first $n$ subsystems. We achieve this by algebraically decomposing the symmetric group $S_{n+1}$. The symmetric group $S_{n+1}$ can be partitioned into $n+1$ disjoint cosets by factoring out the subgroup $S_n$. Specifically, we have,
\begin{equation*}
S_{n+1} = \bigcup_{k=1}^{n+1} (k, n+1) S_n,
\end{equation*}
where $S_{k,n+1}$ denotes the transposition operator exchanging the $k$-th and the $(n+1)$-st tensor factors. 

Using this coset partitioning, the fully symmetric projector $\Pi_{\mathrm{sym}}^{(n+1)}$ (which is the average of all permutation operators) can be recursively expanded as follows,
\begin{align*}
\Pi_{\mathrm{sym}}^{(n+1)} &= \frac{1}{(n+1)!} \sum_{k=1}^{n+1} \sum_{\pi \in S_n} S_{k,n+1} (P_\pi \otimes I_{n+1})\\
&= \frac{1}{n+1} \left( I + \sum_{k=1}^n S_{k,n+1} \right) (\Pi_{\mathrm{sym}}^{(n)} \otimes I_{n+1}).
\end{align*}

After inserting this decomposed projector into the numerator and taking the partial trace $\Tr_{1:n}$, the expression splits cleanly. The identity term passes through directly, yielding $I \cdot \Tr[O_n \Pi_{\mathrm{sym}}^{(n)}]$.

For the transposition terms, we first justify the order of the factors under the partial trace. We use the following elementary cyclicity property of the partial trace over the traced tensor factors: if $A$ is an operator acting on the first $n$ tensor factors and $X$ is an operator on the full $(n+1)$-fold tensor product, then
\begin{equation}
\Tr_{1:n}\!\left[(A\otimes I_{n+1})X\right]
=
\Tr_{1:n}\!\left[X(A\otimes I_{n+1})\right].
\label{eq:partial_trace_cyclicity_first_n}
\end{equation}
Indeed, this follows immediately by testing both sides against arbitrary matrix elements on the remaining $(n+1)$-th tensor factor and applying the usual cyclicity of the full trace on the first $n$ tensor factors.
Applying \eqref{eq:partial_trace_cyclicity_first_n} with
$A=O_n$ and
$X=S_{k,n+1}(\Pi_{\mathrm{sym}}^{(n)}\otimes I_{n+1})$, the $k$-th transposition contribution becomes
\begin{align}
&\Tr_{1:n}\!\left[(O_n\otimes I_{n+1})S_{k,n+1}(\Pi_{\mathrm{sym}}^{(n)}\otimes I_{n+1})\right]
\notag\\
&\qquad=
\Tr_{1:n}\!\left[S_{k,n+1}(\Pi_{\mathrm{sym}}^{(n)}\otimes I_{n+1})(O_n\otimes I_{n+1})\right]
\notag\\
&\qquad=
\Tr_{1:n}\!\left[S_{k,n+1}(\Pi_{\mathrm{sym}}^{(n)}O_n\otimes I_{n+1})\right].
\label{eq:transposition_contribution_to_M}
\end{align}
This justifies defining the cross-correlation operator
\begin{equation}
K_{k,\vec{u}} = \Tr_{1:n} \left[S_{k,n+1} (\Pi_{\mathrm{sym}}^{(n)} O_n \otimes I_{n+1})  \right].\label{k_k,u}
\end{equation}

\subsection{Gram matrices and the marginal permanent}
To evaluate the cross-correlation operators $K_{k,\vec{u}}$, we must explicitly trace out the first $n$ subsystems. We analyze an arbitrary matrix element $\langle x | K_{k,\vec{u}} | y \rangle$ on the $(n+1)$-th subsystem. To perform the partial trace over the observed $n$ systems, we insert a complete orthonormal basis of tensor product states $|j_1, j_2, \dots, j_n\rangle$,
\begin{align*}
&\langle x | K_{k,\vec{u}} | y \rangle \\
=& \sum_{j_1 \dots j_n} \langle j_1 \dots j_n, x | S_{k,n+1}(\Pi_{\mathrm{sym}}^{(n)} O_n \otimes I_{n+1})  | j_1 \dots j_n, y \rangle.
\end{align*} 

Before evaluating the cross-correlation operators $K_{k,\vec{u}}$, we establish a fundamental identity regarding the partial trace of an operator coupled with a transposition.

\begin{fact}[Partial SWAP evaluation]\label{fact:part_swap}
When taking the partial trace over a subsystem $k$ of an $n$-partite operator $A$ multiplied by a transposition $S_{k,n+1}$, the SWAP operator effectively substitutes the external $(n+1)$-th basis states directly into the $k$-th tensor slot of $A$. Formally, we have,
\begin{align*}
&\langle x | \Tr_{1:n} \left[ S_{k,n+1} (A \otimes I_{n+1})  \right] | y \rangle\\
 =& \sum_{j_{\setminus k}} \langle j_1 \dots x \dots j_n | A | j_1 \dots y \dots j_n \rangle,
\end{align*}
where $x$ and $y$ are inserted exactly at the $k$-th position of the bra and ket, respectively, and the sum runs over all basis indices except the $k$-th.
\end{fact}
\begin{proof}
See Appendix \ref{proof_fact:part_swap} for the proof.
\end{proof}

Using this fact, we can rigorously determine the explicit functional form of the operator $K_{k,\vec{u}}$ defined in Eq.~\eqref{k_k,u}.

\begin{lemma}\label{lem:m_k_perm}
For a joint measurement operator $O_n = |u_1 \dots u_n\rangle\langle u_1 \dots u_n|$ and the fully symmetric projector $\Pi_{\mathrm{sym}}^{(n)}$, the cross-correlation operator $K_{k,\vec{u}}$ can be expressed in terms of the permanents of the minors of the Gram matrix $G$ (where $G_{a,b} = \langle u_a | u_b \rangle$),
\begin{equation*}
K_{k,\vec{u}} = \frac{1}{n!} \sum_{j=1}^n \Perm(G(k|j)) |u_j\rangle\langle u_k|,
\end{equation*}
where $G(k|j)$ denotes the minor of $G$ obtained by completely removing the $k$-th row and $j$-th column.
\end{lemma}

\begin{proof}
See Appendix \ref{proof_lem:m_k_perm} for the proof.
\end{proof}

To properly normalize this operator, we must independently evaluate the denominator trace $\Tr[O_n \Pi_{\mathrm{sym}}^{(n)}]$. Expanding the symmetric projector into its permutation sum $\frac{1}{n!} \sum_{\sigma \in S_n} P_\sigma$ and evaluating its expectation over the pure product state $O_n = |u_1 \dots u_n\rangle\langle u_1 \dots u_n|$ yields,
\begin{align*}
\Tr[O_n \Pi_{\mathrm{sym}}^{(n)}] =& \frac{1}{n!} \sum_{\sigma \in S_n} \langle u_1 \dots u_n | P_\sigma | u_1 \dots u_n \rangle \\
=& \frac{1}{n!} \sum_{\sigma \in S_n} \prod_{i=1}^n \langle u_i | u_{\sigma(i)} \rangle.
\end{align*}
Recognizing these inner products as the entries of the Gram matrix $G_{i, \sigma(i)}$, this sum over all permutations in $S_n$ directly reconstructs the full permanent of the Gram matrix,
\begin{equation*}
\Tr[O_n \Pi_{\mathrm{sym}}^{(n)}] = \frac{1}{n!} \Perm(G).
\end{equation*}

Define the fully summed cross-correlation operator by
\begin{equation}
    K_{\vec{u}}
    :=
    \sum_{k=1}^n K_{k,\vec{u}}.
    \label{eq:K_vecu_sum_def}
\end{equation}
This is the unnormalized marginal contribution generated by the measurement correlations. Since the Gram matrix $G$ is Hermitian, its minors satisfy
\begin{equation*}
    \overline{\Perm(G(k\mid j))}
    =
    \Perm(G(j\mid k)).
\end{equation*}
Consequently, the fully summed operator $K_{\vec u}=\sum_k K_{k,\vec u}$ is Hermitian. Its positivity is proved in Lemma~\ref{lem:marginal_psd}. Its trace is
\begin{equation*}
\Tr K_{\vec{u}}
=
\sum_{k=1}^n \Tr\left[  S_{k,n+1} (\Pi_{\mathrm{sym}}^{(n)} O_n \otimes I_{n+1})  \right].
\end{equation*}
By the trace-SWAP identity in Fact~\ref{fact:trace_swap_identity}, applied to
\(A=\Pi_{\mathrm{sym}}^{(n)}O_n\), each individual term satisfies
\[
    \Tr[K_{k,\vec{u}}]
    =
    \Tr[O_n\Pi_{\mathrm{sym}}^{(n)}].
\]
Therefore
\begin{equation}
    \Tr K_{\vec{u}}
    =
    n\Tr[O_n\Pi_{\mathrm{sym}}^{(n)}].
    \label{eq:tr_K_u}
\end{equation}
We now define the trace-normalized empirical marginal state by
\begin{equation}
    \widetilde{\rho}_{\mathrm{marg}}(\vec{u})
    :=
    \frac{K_{\vec{u}}}{n\Tr[O_n\Pi_{\mathrm{sym}}^{(n)}]}
    =
    \frac{1}{n\Tr[O_n\Pi_{\mathrm{sym}}^{(n)}]}
    \sum_{k=1}^n K_{k,\vec{u}}.
    \label{eq:rho_marg_tilde_trace_normalized_def}
\end{equation}
By \eqref{eq:tr_K_u}, this operator satisfies $\Tr[\widetilde{\rho}_{\mathrm{marg}}(\vec{u})]=1$.
Reassembling the full estimator with the expanded denominator trace gives the exact posterior predictive rule
\begin{equation}
\begin{split}
 \widehat\rho_n(\vec u)
 &=\frac{n+1}{n+d}\,
   \frac{I\,\Tr[O_n\Pi_{\rm sym}^{(n)}]+K_{\vec u}}
        {(n+1)\Tr[O_n\Pi_{\rm sym}^{(n)}]}\\
 &=\frac{n\widetilde\rho_{\rm marg}(\vec u)+I}{n+d}.
\end{split}
\label{eq:local_posterior_predictive_state}
\end{equation}
We refer to $\widehat\rho_n(\vec u)$ as the product-Haar posterior predictive state, or simply the local posterior predictive state when the observation model is clear.
This identity is an exact consequence of the Haar posterior.  It separates the outcome-dependent normalized one-particle marginal from the universal identity contribution.  In particular, $\widehat\rho_n(\vec u)\ge I/(n+d)$, so every realized posterior predictive state is full rank.

\begin{remark}[Relation to the arbitrary-effect identity]
After compression of the product Haar-POVM outcome effect to the symmetric
subspace, the arbitrary-effect identity of the companion paper~\cite{CompanionGlobal} gives the abstract
form
\[
    \hat\rho_n(\vec u)=\frac{I+n\rho_{E_{\vec u}}}{n+d}.
\]
The calculation in this section identifies the normalized one-particle marginal
\(\rho_{E_{\vec u}}\) explicitly as
\(\widetilde\rho_{\mathrm{marg}}(\vec u)\).  Thus the abstract identity
and the local permanent calculation have distinct roles: the former determines
the universal identity-regularized form, whereas the latter evaluates the
outcome-dependent marginal itself.  The Gram-minor permanent coefficients are the concrete product-Haar-POVM
realization of this general effect marginal.
\end{remark}

\begin{lemma}[Positivity of the empirical marginal state]\label{lem:marginal_psd}
The fully summed cross-correlation operator $K_{\vec{u}}$ is a positive semi-definite operator. Consequently, its trace-normalized counterpart $\widetilde{\rho}_{\mathrm{marg}}(\vec{u})$ satisfies the requirements of a valid quantum density operator.
\end{lemma}

\begin{proof}
Recall the explicit expression
\begin{equation}
    K_{\vec u}
    =
    \frac{1}{n!}
    \sum_{k=1}^n\sum_{j=1}^n
    \Perm(G(k\mid j))
    \ket{u_j}\bra{u_k}.
    \label{eq:psd_proof_K_explicit_recall}
\end{equation}
Let
\begin{equation}
    V:
    \mathbb C^n\to\mathcal H,
    \qquad
    V\ket{e_i}=\ket{u_i},
    \label{eq:psd_proof_synthesis_map}
\end{equation}
where \(\{\ket{e_i}\}_{i=1}^n\) is the standard basis of \(\mathbb C^n\).
Define the permanent-weight matrix \(C_{\vec u}\) on \(\mathbb C^n\) by
\begin{equation}
    (C_{\vec u})_{j,k}
    :=
    \frac{1}{n!}\Perm(G(k\mid j)).
    \label{eq:psd_proof_C_def}
\end{equation}
Then \eqref{eq:psd_proof_K_explicit_recall} is exactly the factorization
\begin{equation}
    K_{\vec u}
    =
    V C_{\vec u} V^\dagger .
    \label{eq:psd_proof_K_factorization}
\end{equation}
It remains to show that \(C_{\vec u}\ge0\).

For each \(j\), let
\begin{equation}
    \ket{U_{\setminus j}}
    :=
    \bigotimes_{\substack{i=1\\ i\ne j}}^n \ket{u_i},
    \label{eq:psd_proof_leave_one_out_vector}
\end{equation}
where the tensor factors are kept in their original order after deleting the
\(j\)-th entry.  Set
\begin{equation}
    \ket{W_j}
    :=
    \Pi_{\mathrm{sym}}^{(n-1)}\ket{U_{\setminus j}}
    \in \mathcal H^{\otimes(n-1)}.
    \label{eq:psd_proof_projected_leave_one_out}
\end{equation}
Consider the Gram matrix
\begin{equation}
    L_{j,k}
    :=
    \langle W_k\mid W_j\rangle.
    \label{eq:psd_proof_L_def}
\end{equation}
By construction, \(L\) is positive semidefinite.  Expanding the symmetric
projection gives
\begin{align}
    L_{j,k}
    &=
    \left\langle
        U_{\setminus k}
    \middle|
        \Pi_{\mathrm{sym}}^{(n-1)}
    \middle|
        U_{\setminus j}
    \right\rangle 
    \notag\\
    &=
    \frac{1}{(n-1)!}
    \sum_{\substack{\sigma\in S_n:\\ \sigma(k)=j}}
    \prod_{m\ne k}\langle u_m\mid u_{\sigma(m)}\rangle
    \notag\\
    &=
    \frac{1}{(n-1)!}\Perm(G(k\mid j)).
    \label{eq:psd_proof_L_perm_identity}
\end{align}
Here the second line identifies a permutation of the reduced ordered lists with
a permutation \(\sigma\in S_n\) satisfying \(\sigma(k)=j\).  Comparing
\eqref{eq:psd_proof_C_def} and \eqref{eq:psd_proof_L_perm_identity}, we obtain
\begin{equation}
    C_{\vec u}
    =
    \frac{1}{n}L.
    \label{eq:psd_proof_C_equals_L_over_n}
\end{equation}
Therefore \(C_{\vec u}\ge0\).  The factorization
\eqref{eq:psd_proof_K_factorization} then implies \(K_{\vec u}\ge0\).
Finally,
\[
    \widetilde\rho_{\mathrm{marg}}(\vec u)
    =
    \frac{K_{\vec u}}{n\Tr[O_n\Pi_{\mathrm{sym}}^{(n)}]}
\]
has unit trace by \eqref{eq:tr_K_u}.  Hence it is a positive semidefinite
trace-one operator, i.e. a density operator.
\end{proof}

\begin{lemma}[Pointwise covariance of the local posterior predictive state]
\label{lem:pointwise_covariance_quantum_local}
Let \(W\in U(d)\) and write \(W\vec u=(Wu_1,\ldots,Wu_n)\).  Then
\begin{equation}
    \hat\rho_n(W\vec u)
    =
    W\hat\rho_n(\vec u)W^\dagger.
    \label{eq:pointwise_covariance_quantum_local}
\end{equation}
\end{lemma}

\begin{proof}
The Gram matrix of \(W\vec u\) has entries
\(\langle Wu_a\mid Wu_b\rangle=\langle u_a\mid u_b\rangle\), and hence is the
same as the Gram matrix of \(\vec u\).  Therefore all permanent weights in
Lemma~\ref{lem:m_k_perm} are unchanged by the simultaneous rotation of all
outcomes.  The rank-one operators transform as
\begin{equation*}
    \ket{Wu_j}\bra{Wu_k}
    =
    W\ket{u_j}\bra{u_k}W^\dagger.
\end{equation*}
Thus \(K_{W\vec u}=WK_{\vec u}W^\dagger\), and the normalization factor
\(\Tr[O_n\Pi_{\mathrm{sym}}^{(n)}]=\Perm(G)/n!\) is unchanged.  Hence
\begin{equation*}
    \widetilde\rho_{\mathrm{marg}}(W\vec u)
    =
    W\widetilde\rho_{\mathrm{marg}}(\vec u)W^\dagger.
\end{equation*}
Substituting this into
\(\hat\rho_n=(n\widetilde\rho_{\mathrm{marg}}+I)/(n+d)\) gives
\eqref{eq:pointwise_covariance_quantum_local}.
\end{proof}

\subsection{Bayesian prediction and decision-rule minimax optimality for product Haar-POVM data}
\label{subsec:local_bayes_minimax_product_haar_povm_data}
\label{subsec:posterior_relative_entropy_bayes_optimality}
\label{subsec:product_haar_povm_minimax_optimality}

We now record the predictive and decision-theoretic status of the local posterior predictive state in
the product Haar-POVM setting.  Let \(\mathcal D_n\) be the class of all
measurable decision rules
\begin{equation}
    \delta:
    \vec u
    \longmapsto
    \delta(\vec u)\in\mathcal S(\mathcal H)
    \label{eq:local_decision_rule_class_def}
\end{equation}
based on the outcomes of the product Haar-POVM \(E^{(n)}\).  For
\(\delta\in\mathcal D_n\), define the frequentist relative-entropy risk by
\begin{equation}
    R_\delta(\psi)
    :=
    \mathbb E_{\vec u\mid\psi}
    D\!\left(
        \rho_\psi
        \middle\|
        \delta(\vec u)
    \right),
    \label{eq:product_haar_povm_frequentist_risk_def}
\end{equation}
and the corresponding product-Haar-POVM decision-rule minimax risk by
\begin{equation}
    R_n^{\mathrm{mm}}
    :=
    \inf_{\delta\in\mathcal D_n}
    \sup_\psi R_\delta(\psi).
    \label{eq:product_haar_povm_minimax_risk_def}
\end{equation}
This minimax problem is only over decision rules after the independent
product Haar-POVM setting has been fixed; it is not a minimax problem over all
collective quantum measurements.

\begin{proposition}[Posterior predictive state for product Haar-POVM data]
\label{prop:local_posterior_mean_bayes_act}
For product Haar-POVM data and the Haar prior, the unique Bayes act for quantum
relative-entropy loss is the posterior mean
\begin{equation}
    \hat\rho_n(\vec u)
    =
    \mathbb E[\rho_\Phi\mid\vec u].
    \label{eq:local_posterior_mean_bayes_act}
\end{equation}
Equivalently, it has the identity-regularized posterior form
\begin{equation}
    \hat\rho_n(\vec u)
    =
    \frac{n\widetilde\rho_{\mathrm{marg}}(\vec u)+I}{n+d}.
    \label{eq:posterior_mean_evaluated_as_quantum_local_main}
\end{equation}
\end{proposition}

\begin{proof}
The Bayes-act statement is the fixed-observation Bayes identity in Section~\ref{sec:bayes_tools_p2}
applied to the product Haar-POVM \(E^{(n)}\), with the outcome \(x\) equal to
the complete sequence \(\vec u\).  The additional content here is the explicit
Haar-integral and one-particle-marginal calculation, which evaluates that
posterior mean and gives exactly
\eqref{eq:posterior_mean_evaluated_as_quantum_local_main}.
\end{proof}

\begin{corollary}[Decision-rule minimax optimality for product Haar-POVM data]
\label{cor:local_decision_rule_minimax_product_haar_povm_data}
\label{thm:product_haar_povm_minimax_optimality}
Within the class \(\mathcal D_n\) of decision rules based on product
Haar-POVM outcomes, the local posterior predictive state is
minimax for quantum relative-entropy loss:
\begin{equation}
    \sup_\psi R_{\hat\rho_n}(\psi)
    =
    \inf_{\delta\in\mathcal D_n}
    \sup_\psi R_\delta(\psi).
    \label{eq:product_haar_povm_minimax_optimality}
\end{equation}
\end{corollary}

\begin{proof}
By Lemma~\ref{lem:product_haar_povm_covariance}, the product Haar-POVM model is
covariant.  By Lemma~\ref{lem:pointwise_covariance_quantum_local}, the local posterior predictive state is pointwise covariant.  Proposition~\ref{prop:local_posterior_mean_bayes_act}
shows that it attains the Haar-Bayes value among all decision rules in
\(\mathcal D_n\).

The minimax conclusion is the specialization to decision rules based on the product Haar-POVM outcomes of the
Hunt--Stein averaging argument in the compact-group averaging argument: for a
covariant observation model and a covariant decision rule, the frequentist risk
is constant on the unitary orbit of pure states, hence its Haar-Bayes risk equals
its worst-case risk.  Therefore the Haar-Bayes value and the minimax value
coincide within \(\mathcal D_n\), and the local posterior predictive state is
decision-rule minimax.
\end{proof}

\begin{remark}[Scope of the local minimax statement]
The corollary compares only decision rules that map the observed product Haar-POVM outcome sequence to a density matrix.  The product measurement is specified in advance.  Thus this is not a minimax statement over all quantum
measurements on \(\rho_\psi^{\otimes n}\), and it does not include collective
measurements or other POVMs in the optimization.
\end{remark}

\section{Finite-Sample Risk and Sample Complexity}
\label{sec:finite_sample_risk_complexity}
This section concerns the exact posterior predictive state for the original
product Haar observation.  It is independent of the randomized-mesh MLE
analysis in Section~\ref{sec:rmesh_main}.  In particular, all bounds below hold
for every finite $n$ and $d$, without choosing a mesh or taking an asymptotic
limit.

\subsection{Hilbert--Schmidt projection and overlap}
Let $\Phi$ have the Haar prior, let $\vec U$ be the product Haar outcome, and
write
\[
 \widehat\rho_n=\mathbb E[\rho_\Phi\mid\vec U].
\]
For any square-integrable Hermitian matrix-valued estimator $T(\vec U)$, the
conditional-expectation Pythagorean identity gives
\begin{equation}
 \mathbb E\|\rho_\Phi-T\|_2^2
 =\mathbb E\|\rho_\Phi-\widehat\rho_n\|_2^2
 +\mathbb E\|\widehat\rho_n-T\|_2^2.
 \label{eq:hs_pythagorean_posterior}
\end{equation}
Thus the posterior mean minimizes Bayesian Hilbert--Schmidt risk over all
square-integrable matrix-valued estimators measurable with respect to the
observed data.  The comparison class here is not restricted to density
matrices.

The local posterior overlap, equivalently the mean posterior purity, is
\begin{equation}
 \mathcal P_{n,d}^{\rm loc}
 :=\mathbb E\operatorname{Tr}(\widehat\rho_n^2)
 =\mathbb E\operatorname{Tr}(\rho_\Phi\widehat\rho_n).
 \label{eq:finite_local_overlap_purity_identity}
\end{equation}
Since $\rho_\Phi$ is pure,
\begin{equation}
 \mathbb E\|\rho_\Phi-\widehat\rho_n\|_2^2
 =1-\mathcal P_{n,d}^{\rm loc}.
 \label{eq:hs_risk_equals_one_minus_purity}
\end{equation}
This posterior-overlap quantity is distinct from the fidelity of a reported
ray.

\subsection{Empirical linear inversion and its exact variance}
Define the empirical linear-inversion estimator
\begin{equation}
 \overline L_n
 :=\frac1n\sum_{i=1}^n\bigl((d+1)\rho_{U_i}-I\bigr).
 \label{eq:linear_inversion_estimator}
\end{equation}
It is used only as a matrix-valued competitor in
\eqref{eq:hs_pythagorean_posterior}; it need not be positive.  The estimator is
covariant under simultaneous unitary rotations.  To verify unbiasedness and
its variance, let $U$ have one-copy density
$d\operatorname{Tr}(\rho_\phi\rho_u)$ with respect to normalized Haar measure.
The second Haar moment
\[
 \int \rho_u^{\otimes2}\,du=\frac{I+S}{d(d+1)}
\]
implies
\begin{equation}
 \mathbb E_\phi\rho_U=\frac{I+\rho_\phi}{d+1},
 \qquad
 \mathbb E_\phi\bigl((d+1)\rho_U-I\bigr)=\rho_\phi.
 \label{eq:linear_inversion_unbiasedness}
\end{equation}
Put $L(U)=(d+1)\rho_U-I$.  Since $\rho_U^2=\rho_U$,
\[
 L(U)^2=(d^2-1)\rho_U+I,
 \qquad
 \operatorname{Tr}L(U)^2=d^2+d-1.
\]
Therefore, using $\operatorname{Tr}(\rho_\phi^2)=1$ and independence,
\begin{equation}
\begin{split}
 \mathbb E_\phi
 \bigl\|\overline L_n-\rho_\phi\bigr\|_2^2
 &=\frac{1}{n}\,
   \mathbb E_\phi
   \bigl\|L(U)-\rho_\phi\bigr\|_2^2\\
 &=\frac{d^2+d-2}{n}.
\end{split}
\label{eq:linear_inversion_exact_hs_risk}
\end{equation}
Thus the exact numerator is $d^2+d-2$.

Combining \eqref{eq:hs_pythagorean_posterior},
\eqref{eq:hs_risk_equals_one_minus_purity}, and
\eqref{eq:linear_inversion_exact_hs_risk} yields the finite-sample result.
\begin{theorem}[Finite-sample local posterior-overlap bound]
\label{thm:finite_sample_local_overlap}
For every $n\ge1$ and $d\ge2$,
\begin{equation}
 \mathcal P_{n,d}^{\rm loc}
 \ge 1-\frac{d^2+d-2}{n}.
 \label{eq:finite_sample_local_overlap_bound}
\end{equation}
When the displayed right-hand side is negative the inequality is, of course,
weaker than the trivial bound $\mathcal P_{n,d}^{\rm loc}\ge0$.
\end{theorem}

\subsection{Explicit relative-entropy bound}
The exact posterior satisfies the pointwise spectral floor
$\widehat\rho_n\ge I/(n+d)$.  Set $a=1/(n+d)$.  For every
$x\in[a,1]$, the chord bound for the convex function $-\log x$ gives
\[
 -\log x\le\frac{1-x}{1-a}\log\frac1a.
\]
Multiplying by $x$, summing over the eigenvalues of $\widehat\rho_n$, and
averaging yields
\begin{equation}
 R_{n,d}^{\rm loc}:=\mathbb E S(\widehat\rho_n)
 \le\frac{1-\mathcal P_{n,d}^{\rm loc}}{1-1/(n+d)}\log(n+d).
 \label{eq:entropy_from_purity_and_floor}
\end{equation}
Hence Theorem~\ref{thm:finite_sample_local_overlap} gives
\begin{theorem}[Finite-sample relative-entropy risk]
\label{thm:finite_sample_relative_entropy}
For every $n\ge1$ and $d\ge2$,
\begin{equation}
\begin{split}
 R_{n,d}^{\rm loc}
 \le \frac{d^2+d-2}{n}\,
       \frac{n+d}{n+d-1}
\log(n+d).
\end{split}
\label{eq:finite_sample_relative_entropy_bound}
\end{equation}
\end{theorem}

\subsection{Finite-sample count bound}
The preceding theorem immediately gives a checkable sufficient condition.  The
following convenient closed form is not optimized.
\begin{corollary}[Explicit finite-sample sufficient condition]
\label{cor:finite_sample_complexity}
Let $d\ge2$ and $0<\epsilon\le1$, and set $C_d=d^2+d-2$.  If
\begin{equation}
 n\ge \frac{4C_d}{\epsilon}
 \log\!\left(\frac{4C_d(d+1)}{\epsilon}\right).
\label{eq:explicit_finite_sample_condition}
\end{equation}
then $R_{n,d}^{\rm loc}\le\epsilon$.  In particular, this gives the rigorous
finite-sample scaling
\begin{equation}
 n=O\!\left(\frac{d^2}{\epsilon}
 \log\frac{d}{\epsilon}\right).
 \label{eq:quadratic_d_sample_complexity}
\end{equation}
\end{corollary}
\begin{proof}
Since $(n+d)/(n+d-1)\le2$, it is enough that
$2C_d\log(n+d)/n\le\epsilon$.  Put
$A=2C_d/\epsilon$.  The right-hand side of
\eqref{eq:explicit_finite_sample_condition} is $2A\log(2A(d+1))$.
Let $n_0=2A\log(2A(d+1))$.  Since $A\ge4$ and $d\ge2$, we have
$n_0\ge A$.  The function
$x\mapsto x/A-\log(x+d)$ is increasing for $x\ge A$, because its derivative
is $A^{-1}-(x+d)^{-1}>0$.  Moreover,
\begin{align*}
 \frac{n_0}{A}=2\log(2A(d+1))
& \ge \log\bigl(2A\log(2A(d+1))+d\bigr)\notag\\
 &=\log(n_0+d),
\end{align*}
where the inequality follows from
$2A(d+1)^2\ge 2A\log(2A(d+1))+d$, valid for $A\ge4$ and $d\ge2$.
Hence $\log(n+d)\le n/A$ for every $n\ge n_0$, and the desired inequality
follows.
\end{proof}

The quadratic dependence on $d$ comes from comparing the posterior mean with
the simple linear-inversion estimator.  It is a uniform finite-sample
guarantee.  By contrast, the randomized-mesh argument below is Fisher efficient
at each fixed mesh and has coefficient $a_k\to d-1$ only after the iterated
large-sample and mesh-refinement limits.  It yields the sharper fixed-dimensional
linear-$d$ scale $(d/\epsilon)\log(d/\epsilon)$ asymptotically, but not a uniform
linear-$d$ nonasymptotic theorem.  The two statements therefore concern
different strengths of guarantee and are not contradictory.

\section{Recovering Finite Meshes and the Shared-Orientation Experiment}
\label{sec:rmesh_main}
We first explain why the following finite-mesh regularization is needed.  The
continuous covariant Haar measurement has exactly the desired Fisher
information, namely $2g_{\rm FS}$, as discussed in
\cite[Sec.~5]{Hayashi1998PureStateEstimation}.  It is nevertheless not a
regular model to which the standard MLE efficiency theory can be applied
directly.  Although the density
$d\operatorname{Tr}(\rho_\theta\rho_v)$ has full measure-theoretic support,
it vanishes on the parameter-dependent projective hyperplane orthogonal to
$\theta$.  Consequently, the density is not uniformly bounded away from zero,
and the log likelihood and its derivatives can be singular near that zero set.
Thus the usual common-positive-density regularity conditions for the smooth
log-likelihood treatment are unavailable; compare
\cite[Secs.~2.1 and 4.3--4.4]{AmariNagaoka2000InformationGeometry}.  The finite
mesh replaces the continuous outcome by finitely many positive-probability
cells.  At each sufficiently fine fixed mesh it therefore restores a common
positive support and a regular finite-alphabet likelihood, while refinement
recovers the Fisher information of the original continuous covariant
measurement.

The randomized-mesh construction is an auxiliary comparison device, not a
replacement for the original product-Haar prediction problem.  Its role is to
produce a covariant ray-valued coarse observation whose risk can be analyzed by
ordinary finite-alphabet MLE theory.  The resulting coarse posterior gives an
upper bound on the entropy, and a lower bound on the purity, of the posterior
conditioned on the complete product-Haar outcomes.  The opposite bounds come
from the exact collective benchmarks.  Thus the argument has three distinct
levels: the fine continuous observation, the shared-orientation mesh
experiment, and the reported-ray coarse observation.

The randomized-mesh construction is used only to obtain the upper sides of the
sharp asymptotic comparisons.  Its quantifiers are therefore organized as
follows.  We first fix a sufficiently fine finite mesh and let $n\to\infty$ in
the resulting regular finite-alphabet model.  Only after this fixed-mesh limit
has been taken do we refine the mesh.  No triangular choice $k=k_n$ is used.

The only mesh properties needed below are full support, global identifiability,
and uniform recovery of the fine Fisher information.  Stating these properties
at the mesh level separates the approximation argument from the standard
fixed-mesh MLE argument.
Let $\Theta=\mathbb P(\mathbb C^d)$, with invariant probability measure $\nu$
and Fubini--Study metric $g_{\rm FS}$.  Throughout the mesh construction, we allow a partition to leave out a set
of invariant measure zero.  Thus its cells are pairwise disjoint and their
union agrees with $\Theta$ except on a set $N$ satisfying $\nu(N)=0$.  Since
every model distribution $P_\theta$ has a density with respect to $\nu$, we
also have $P_\theta(N)=0$ for every $\theta$.  The cell label may therefore be
defined arbitrarily on $N$, without changing any cell probability, score,
conditional expectation, or Fisher information.  Choose nested finite Borel
cell families in this sense,
\[
 \mathcal C_k=\{C_{k,1},\ldots,C_{k,m_k}\}
\]
with cells of positive $\nu$-measure whose boundaries have $\nu$-measure zero.  Put
\begin{equation}
 A_{k,j}:=\int_{C_{k,j}}\rho_v\,\nu(dv),
 \qquad
 q_{\theta,k}(j):=d\,\operatorname{Tr}(\rho_\theta A_{k,j}).
 \label{eq:recovering_mesh_model}
\end{equation}
We impose the following conditions.

\begin{definition}[Recovering mesh sequence]
\label{def:recovering_mesh}
A nested sequence $(\mathcal C_k)_k$ of finite Borel cell families, in the
above measure-zero convention, is called recovering if:
\begin{enumerate}
\item the cells are pairwise disjoint, cover $\Theta$ up to a set of
$\nu$-measure zero, and each cell has positive $\nu$-measure and boundary of
$\nu$-measure zero;
\item the maximal Fubini--Study diameter tends to zero;
\item for every sufficiently large $k$, the map
$\theta\mapsto q_{\theta,k}$ is globally identifiable; and
\item writing $I_{\theta,k}$ for the Fisher information of
\eqref{eq:recovering_mesh_model},
\begin{equation}
 \eta_k:=
 \sup_{\substack{\theta\in\Theta,\ v\in T_\theta\Theta\\
 g_{{\rm FS},\theta}[v,v]=1}}
 \left|I_{\theta,k}[v,v]-2\right|
 \longrightarrow0.
 \label{eq:uniform_fisher_recovery}
\end{equation}
\end{enumerate}
\end{definition}

Global identifiability in item 3 is guaranteed when the cell operators
$\{A_{k,j}\}_j$ span the Hermitian matrices.  It is enough to establish this
spanning condition at one mesh, because every parent-cell operator is the sum
of its child-cell operators.

For clarity, Appendix~\ref{app:informationally_complete_meshes} proves the
existence of recovering sequences in two separate stages.  First, it constructs
an informationally complete finite partition and then refines it through nested
partitions of vanishing diameter whose cells have positive measure and boundaries of measure zero; the spanning
property, and hence global identifiability, persist under refinement.  Second,
it writes each directional mesh score as the conditional expectation of the
fine score.  Martingale convergence gives pointwise Fisher recovery, while
monotonicity and Dini's theorem upgrade this to uniform recovery on the unit
tangent bundle.  Proposition~\ref{prop:recovering_meshes_exist} records the
result.

Condition \eqref{eq:uniform_fisher_recovery} is equivalently the coordinate-free
quadratic-form bound
\begin{equation}
 (2-\eta_k)g_{{\rm FS},\theta}
 \le I_{\theta,k}
 \le(2+\eta_k)g_{{\rm FS},\theta}.
 \label{eq:fisher_recovery_form_bound}
\end{equation}
Consequently, there is $k_0$ such that every $k\ge k_0$ has uniformly
nondegenerate Fisher information.  Thus Fisher nondegeneracy need not be
proved separately from informational completeness.

The remaining argument uses the mesh through four separate implications:
\begin{equation}
 \begin{array}{c}
 \text{uniform Fisher recovery}\\
 \Downarrow\\
 \text{fixed-mesh MLE efficiency}\\
 \Downarrow\\
 \text{covariant reported-ray risk}\\
 \Downarrow\\
 \text{coarse-posterior entropy and purity bounds}.
 \end{array}
 \label{eq:mesh_bridge_map}
\end{equation}
The first implication is treated in Section~\ref{sec:fixed_mesh_mle_efficiency};
the last two are treated in
Section~\ref{sec:posterior_comparison_asymptotics}.  This display is only a map
of the proof: no additional limit or statistical assumption is introduced.

For a fixed $k\ge k_0$, draw one Haar unitary $G$ before observing the sample
and retain it for all copies.  If the fine outcomes are $U_1,\ldots,U_n$, report
\begin{equation}
 Y_{n,k}=(G,J_1,\ldots,J_n),
 \qquad G^{-1}U_i\in C_{k,J_i}.
 \label{eq:shared_orientation_observation}
\end{equation}
Conditionally on $G=g$, the reference parameter is $\theta=g^{-1}\psi$ and the
labels are i.i.d. with law $q_{\theta,k}$.  Let
$\widehat\theta_{n,k}$ be one fixed measurable global-MLE selector for this
finite-alphabet model and define the reported physical ray by
\begin{equation}
 \widehat\Psi_{n,k}:=G\widehat\theta_{n,k}.
 \label{eq:shared_orientation_reported_ray}
\end{equation}
Finite-sample uniqueness is not assumed.  The rule is exactly covariant because
$(G,J^n)\mapsto(HG,J^n)$ under a physical rotation $H$.

\section{Fixed-Mesh Efficiency and Haar-Averaged Risk}
\label{sec:fixed_mesh_mle_efficiency}
Fix $k\ge k_0$.  Each probability in \eqref{eq:recovering_mesh_model} is
strictly positive: the zero set of
$v\mapsto\operatorname{Tr}(\rho_\theta\rho_v)$ is a of Haar measure zero projective
hyperplane and cannot contain a positive-measure cell.  Continuity, compactness,
and the finite alphabet therefore give
\begin{equation}
 c_k:=\min_{\theta,j}q_{\theta,k}(j)>0.
 \label{eq:fixed_mesh_uniform_support}
\end{equation}
The three inputs to fixed-mesh regularity are now explicit.  Positive cell
probabilities and compactness give the uniform support bound
\eqref{eq:fixed_mesh_uniform_support}; the spanning construction in
Appendix~\ref{app:informationally_complete_meshes} gives global
identifiability; and uniform Fisher recovery gives the nondegeneracy bound
\eqref{eq:fisher_recovery_form_bound}.  Together they make the fixed-$k$ model
a compact, smooth, identifiable, full-support regular finite-alphabet model.

Define
\begin{equation}
 a_k(\theta):=
 \operatorname{Tr}\!\left(g_{{\rm FS},\theta}I_{\theta,k}^{-1}\right),
 \qquad
 \overline a_k:=\int_\Theta a_k(\theta)\,\nu(d\theta).
 \label{eq:averaged_fixed_mesh_coefficient}
\end{equation}

\begin{theorem}[Fixed-mesh MLE efficiency in the form needed here]
\label{thm:fixed_mesh_selected_mle}
For every fixed $k\ge k_0$ and every $\theta\in\Theta$,
\begin{equation}
 \sqrt n\,\operatorname{Log}_\theta\widehat\theta_{n,k}
 \Longrightarrow N(0,I_{\theta,k}^{-1}),
 \label{eq:selected_mle_asymptotic_normality}
\end{equation}
and
\begin{equation}
 n\,\mathbb E_\theta
 \bigl[1-F(\widehat\theta_{n,k},\theta)\bigr]
 \longrightarrow a_k(\theta).
 \label{eq:fixed_mesh_pointwise_risk}
\end{equation}
Moreover, for some finite constant $C_k$,
\begin{equation}
 \sup_{\theta\in\Theta}
 \mathbb E_\theta
 \bigl[1-F(\widehat\theta_{n,k},\theta)\bigr]
 \le \frac{C_k}{n}
 \label{eq:fixed_mesh_risk_domination}
\end{equation}
for all sufficiently large $n$.  Consequently,
\begin{equation}
 n\int_\Theta
 \mathbb E_\theta\bigl[1-F(\widehat\theta_{n,k},\theta)\bigr]
 \,\nu(d\theta)
 \longrightarrow\overline a_k.
 \label{eq:fixed_mesh_haar_averaged_risk}
\end{equation}
\end{theorem}

\begin{proof}
The preceding regularity statement places the fixed-$k$ model directly within
ordinary finite-alphabet MLE theory.  It gives consistency and
\eqref{eq:selected_mle_asymptotic_normality}.  In Fubini--Study normal
coordinates,
\[
 1-F(\operatorname{Exp}_\theta v,\theta)
 =g_{{\rm FS},\theta}[v,v]+O(\|v\|^4).
\]
The standard quadratic-risk expansion therefore gives
\eqref{eq:fixed_mesh_pointwise_risk}.

For completeness, only a domination estimate, not a parameter-uniform risk
expansion, is needed.  Compact identifiability gives a uniform
Kullback--Leibler separation outside any fixed neighborhood of the true
parameter.  Finite-alphabet concentration then localizes every selected global
MLE to that neighborhood with probability at least $1-Ce^{-cn}$, uniformly in
$\theta$.  Within the neighborhood, the score expansion, bounded scores from
\eqref{eq:fixed_mesh_uniform_support}, and the uniform inverse-Fisher bound from
\eqref{eq:fisher_recovery_form_bound} give mean squared geodesic error
$O_k(n^{-1})$.  Since infidelity is locally bounded by a constant times squared
Fubini--Study distance and is globally at most one,
\eqref{eq:fixed_mesh_risk_domination} follows.  Dominated convergence now gives
\eqref{eq:fixed_mesh_haar_averaged_risk}.
\end{proof}

The shared orientation converts this Haar-averaged reference risk into a
pointwise physical risk.  For every fixed physical ray $\psi$, the reference
parameter $G^{-1}\psi$ is Haar distributed and unitary invariance gives
\[
 F(\widehat\Psi_{n,k},\psi)
 =F(\widehat\theta_{n,k},G^{-1}\psi).
\]
Hence \eqref{eq:fixed_mesh_haar_averaged_risk} implies
\begin{equation}
 n\,\mathbb E_\psi
 \bigl[1-F(\widehat\Psi_{n,k},\psi)\bigr]
 \longrightarrow\overline a_k
 \label{eq:averaged_mle_infidelity_limit}
\end{equation}
for every $\psi$.  No uniform risk expansion in the reference parameter is
required; the exact physical constancy is supplied by the shared Haar
orientation.

\section{Posterior Comparison and Sharp Asymptotics}
\label{sec:posterior_comparison_asymptotics}
\subsection{Recovery of the fine coefficient}
\label{subsec:fine_coefficient_recovery}
The uniform Fisher-recovery condition immediately controls the inverse Fisher
information.  From \eqref{eq:fisher_recovery_form_bound},
\begin{equation}
 \frac{1}{2+\eta_k}g_{{\rm FS},\theta}^{-1}
 \le I_{\theta,k}^{-1}
 \le\frac{1}{2-\eta_k}g_{{\rm FS},\theta}^{-1}.
 \label{eq:inverse_fisher_recovery_bound}
\end{equation}
Since $\dim_{\mathbb R}\Theta=2(d-1)$,
\begin{equation}
 \frac{2(d-1)}{2+\eta_k}
 \le a_k(\theta)
 \le\frac{2(d-1)}{2-\eta_k}.
 \label{eq:coefficient_recovery_bound}
\end{equation}
Thus $a_k(\theta)\to d-1$ uniformly in $\theta$, and in particular
\begin{equation}
 \overline a_k\longrightarrow d-1.
 \label{eq:averaged_coefficient_limit}
\end{equation}

\subsection{Covariant coarse-graining of the posterior}
\label{subsec:covariant_posterior_comparison}
Let $X^n$ denote the complete product-Haar outcomes and set
$Z_{n,k}:=\widehat\Psi_{n,k}$.  The information order used below is
\begin{align*}
 \text{complete Haar outcomes}
& \longrightarrow \text{shared-orientation cell labels} \\
 & \longrightarrow \text{selected MLE ray}.
\end{align*}
More precisely, the relevant coarse-graining is the Markov chain
\[
 (X^n,G)\longrightarrow (G,J_1,\ldots,J_n)
 \longrightarrow Z_{n,k}.
\]
Each arrow discards information.  Consequently the posterior entropy can only
increase and the posterior purity can only decrease.  The MLE calculation is
therefore used only to control a coarse posterior, which in turn bounds the
posterior conditioned on the complete product-Haar data.
Adding the independent variable $G$ does not change the fine posterior.  Define
the posterior conditioned only on the reported ray by
\[
 \mu_{n,k}(Z):=\mathbb E[\rho_\Psi\mid Z_{n,k}=Z].
\]
Exact covariance and stabilizer invariance imply
\begin{equation}
 \mu_{n,k}(Z)
 =r_{n,d,k}\rho_Z
 +\frac{1-r_{n,d,k}}{d-1}(I-\rho_Z),
 \label{eq:shared_mesh_depolarizing_posterior}
\end{equation}
where
\begin{equation}
 r_{n,d,k}:=\mathbb E F(Z_{n,k},\Psi)
 =1-\frac{\overline a_k}{n}+o_k(n^{-1})
 \label{eq:shared_mesh_fidelity_expansion}
\end{equation}
by \eqref{eq:averaged_mle_infidelity_limit}.  Therefore
\begin{align}
 h_d(r_{n,d,k})
 &=\overline a_k\frac{\log n}{n}
   +o_k\!\left(\frac{\log n}{n}\right),
 \label{eq:shared_mesh_entropy_endpoint}\\
 p_d(r_{n,d,k})
 &=1-\frac{2\overline a_k}{n}+o_k(n^{-1}).
 \label{eq:shared_mesh_purity_endpoint}
\end{align}
The posterior tower and the entropy/purity coarse-graining identities give
\begin{align}
 h_d\!\left(\frac{n+1}{n+d}\right)
 &\le R_{n,d}^{\rm loc}\le h_d(r_{n,d,k}),
 \label{eq:shared_mesh_entropy_sandwich}\\
 \frac{(n+1)^2+d-1}{(n+d)^2}
 &\ge\mathcal P_{n,d}^{\rm loc}\ge p_d(r_{n,d,k}).
 \label{eq:shared_mesh_purity_sandwich}
\end{align}

\subsection{Sharp asymptotics through iterated limits}
\label{subsec:iterated_limit_asymptotics}
Define
\[
 A_{n,d}:=\frac{n}{\log n}R_{n,d}^{\rm loc},
 \qquad
 B_{n,d}:=\frac n2(1-\mathcal P_{n,d}^{\rm loc}).
\]
For every fixed $k\ge k_0$, the upper sides of
\eqref{eq:shared_mesh_entropy_sandwich}--\eqref{eq:shared_mesh_purity_sandwich}
give
\[
 \limsup_{n\to\infty}A_{n,d}\le\overline a_k,
 \qquad
 \limsup_{n\to\infty}B_{n,d}\le\overline a_k.
\]
Only now is the mesh refined.  By \eqref{eq:averaged_coefficient_limit},
\[
 \limsup_{n\to\infty}A_{n,d}\le d-1,
 \qquad
 \limsup_{n\to\infty}B_{n,d}\le d-1.
\]
The exact collective endpoints in the left sides of the sandwich bounds give
the reverse liminf inequalities.  Consequently,
\begin{equation}
 R_{n,d}^{\rm loc}
 =(d-1)\frac{\log n}{n}
 +o\!\left(\frac{\log n}{n}\right),
 \label{eq:sharp_local_entropy_asymptotic}
\end{equation}
and
\begin{equation}
 \mathcal P_{n,d}^{\rm loc}
 =1-\frac{2(d-1)}{n}+o(n^{-1}).
 \label{eq:sharp_local_purity_asymptotic}
\end{equation}
The order of limits remains essential: first fix any sufficiently fine mesh and
let $n\to\infty$, and only afterward let $k\to\infty$.  No triangular schedule
is used.

\subsection{Finite-sample numerical illustration}
\label{subsec:numerical_illustration}
The preceding asymptotic conclusions are entirely analytic and do not rely on
numerical evidence.  We conclude this section with a finite-sample illustration
of the gap between the exact product-Haar posterior and the collective
benchmark.  The supplied Monte Carlo data contain $50{,}000$ samples for each displayed pair $(n,d)$ with $1\le n\le8$ and $d\in\{2,3,4,8,16\}$. They evaluate the exact permanent-based local posterior overlap and the exact collective benchmark. Because these small-$n$ computations concern the exact local posterior rather than the randomized-mesh MLE, they are used only as finite-sample illustrations and not as evidence that an asymptotic regime has been reached. Posterior overlap, posterior purity, and reported-ray fidelity are kept distinct.

\begin{figure*}[t]
 \centering
 \includegraphics[width=0.96\textwidth]{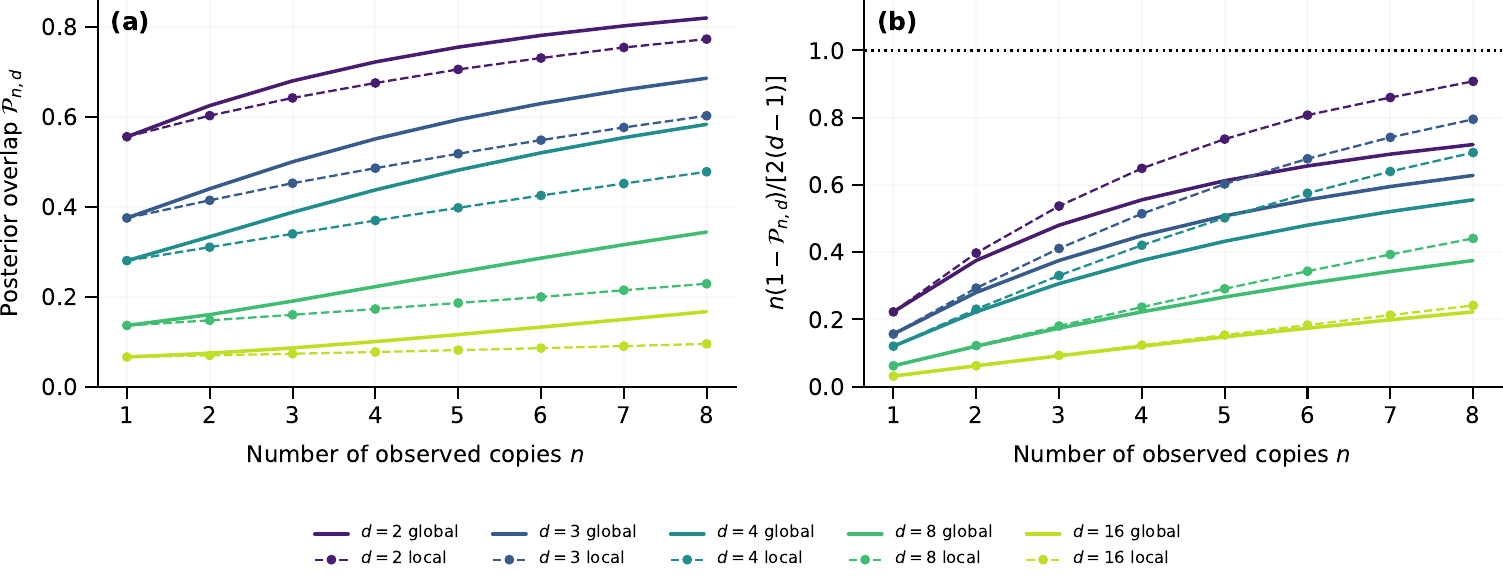}
 \caption{Finite-sample comparison of the exact global collective posterior-overlap endpoint and the Monte Carlo estimate of the local product-Haar posterior overlap. The local quantity equals the mean posterior purity $\mathcal P_{n,d}^{\rm loc}=\mathbb E\operatorname{Tr}(\widehat\rho_n^2)$ and is not a reported-ray fidelity. The displayed values use $50{,}000$ samples for each pair $(n,d)$.}
 \label{fig:global_local_posterior_overlap}
\end{figure*}

Figure~\ref{fig:global_local_posterior_overlap} displays the finite-sample gap between the product-Haar posterior and the unrestricted collective benchmark. It is consistent with the exact ordering proved above, while the rigorous finite-sample risk and sample-count guarantees remain those of Section~\ref{sec:finite_sample_risk_complexity}.

\section{Discussion}
Independent one-copy Haar measurements admit an exact full-rank posterior predictive state whose coefficients are governed by Gram-minor permanents. This solves the density-matrix decision problem after the product observation has been fixed and gives decision-rule minimaxity within that experiment. The unrestricted collective optimization remains the subject of the companion paper and enters here only through its two exact endpoints.

The randomized mesh is organized around a shared orientation.  Uniform Fisher recovery makes every sufficiently fine fixed mesh a regular finite-alphabet model, while the shared Haar orientation converts its Haar-averaged MLE risk into an exactly covariant physical risk.  Thus only the ordinary asymptotic efficiency of one selected MLE and an $O(n^{-1})$ risk domination are needed.  This avoids both a parameter-uniform risk expansion and the stronger task of controlling every element of the finite-sample argmax correspondence.

Two sample-count statements must be distinguished. The linear-inversion comparison proves a uniform finite-sample $O((d^2/\epsilon)\log(d/\epsilon))$ guarantee. The randomized-mesh argument first fixes $k$ and lets $n\to\infty$, obtaining limsup bounds with coefficient $\overline a_k$; only afterward is the mesh refined, giving $\overline a_k\to d-1$. This yields the sharper fixed-dimensional asymptotic scale $(d/\epsilon)\log(d/\epsilon)$ without a triangular schedule. The remaining directions include efficient permanent evaluation and reproducible randomized-mesh numerics.

\appendices

\section{Construction and Fisher Recovery of the Meshes}
\label{app:informationally_complete_meshes}

This appendix proves that the recovering meshes used in
Section~\ref{sec:rmesh_main} exist.  The proof has three steps.  We first build
one finite partition whose cell operators span the Hermitian matrices.  We then
refine it through nested partitions of vanishing diameter whose cells have
positive measure and boundaries of measure zero, and verify that informational
completeness is preserved.
Finally, we identify every directional mesh score as a conditional expectation
of the fine score.  This gives pointwise Fisher recovery by martingale
convergence and uniform recovery by Dini's theorem.

\begin{proposition}[Existence of recovering meshes]
\label{prop:recovering_meshes_exist}
There is a nested sequence of finite Borel cell families
$\mathcal C_k=\{C_{k,1},\ldots,C_{k,m_k}\}$ in the measure-zero convention
stated in Section~\ref{sec:rmesh_main}, on
$\Theta=\mathbb P(\mathbb C^d)$ such that:
\begin{enumerate}
\item the cells cover $\Theta$ up to a set of $\nu$-measure zero, and
every cell has positive $\nu$-measure and boundary of $\nu$-measure zero;
\item the maximal Fubini--Study diameter tends to zero;
\item for every $k$, the cell operators
\[
 A_{k,j}=\int_{C_{k,j}}\rho_v\,\nu(dv)
\]
span the real vector space of Hermitian matrices, and hence
$\theta\mapsto q_{\theta,k}$ is globally identifiable; and
\item the Fisher information satisfies
\[
 \sup_{\substack{\theta\in\Theta,\ w\in T_\theta\Theta\\
 g_{{\rm FS},\theta}[w,w]=1}}
 \left|I_{\theta,k}[w,w]-2\right|\longrightarrow0.
\]
\end{enumerate}
In particular, $(\mathcal C_k)_k$ is a recovering mesh sequence in the sense of
Definition~\ref{def:recovering_mesh}.
\end{proposition}

\subsection{An informationally complete initial partition}

We begin with a concrete finite spanning family of rays.  Fix an orthonormal
basis $\{e_a\}_{a=1}^d$ of $\mathbb C^d$.  Consider the $d^2$ rays represented by
\[
 e_a,\qquad
 \frac{e_a+e_b}{\sqrt2},\qquad
 \frac{e_a+\mathrm i e_b}{\sqrt2},
 \qquad 1\le a<b\le d.
\]
Their rank-one projectors span the Hermitian matrices: the projectors onto
$e_a$ give the diagonal matrix units, while the real and imaginary
superpositions give the symmetric and antisymmetric off-diagonal Hermitian
matrix units after subtracting the corresponding diagonal terms.

\begin{lemma}[Stability of spanning under small cell averages]
\label{lem:small_cell_spanning}
Let $\xi_1,\ldots,\xi_{d^2}$ be rays whose projectors span the Hermitian
matrices.  There are pairwise disjoint Borel neighborhoods
$B_1,\ldots,B_{d^2}$, each of positive $\nu$-measure and with boundary of
$\nu$-measure zero, such that the
operators
\[
 A_\ell:=\int_{B_\ell}\rho_v\,\nu(dv),
 \qquad 1\le \ell\le d^2,
\]
span the Hermitian matrices.
\end{lemma}

\begin{proof}
Choose pairwise disjoint sufficiently small geodesic balls centered at the
distinct rays $\xi_\ell$, with radii chosen so that their boundaries have zero
invariant measure.  Continuity of $v\mapsto\rho_v$ gives
\[
 \frac{A_\ell}{\nu(B_\ell)}\longrightarrow\rho_{\xi_\ell}
\]
as the radius of $B_\ell$ tends to zero.  In a finite-dimensional vector
space, the set of ordered $d^2$-tuples that form a basis is open.  After
shrinking the balls, the normalized operators
$A_\ell/\nu(B_\ell)$ therefore still span the Hermitian matrices.  Multiplying
each member by the positive scalar $\nu(B_\ell)$ does not change their linear
span.
\end{proof}

Shrink the balls further, if necessary, so that
$\nu(\bigcup_{\ell=1}^{d^2}B_\ell)<1$, and set
\[
 C_{1,0}:=\Theta\setminus\bigcup_{\ell=1}^{d^2}B_\ell.
\]
Then $C_{1,0}$ has positive measure, and its boundary is contained in the finite
union of the boundaries $\partial B_\ell$, each of which has $\nu$-measure zero.  Thus
\[
 \mathcal C_1:=\{C_{1,0},B_1,\ldots,B_{d^2}\}
\]
is a genuine finite Borel partition of $\Theta$ whose cells all have positive
$\nu$-measure and boundaries of $\nu$-measure zero.  Its cell operators contain the spanning subfamily
constructed in Lemma~\ref{lem:small_cell_spanning}; hence all cell operators of
$\mathcal C_1$ together span the Hermitian matrices.

\subsection{Nested refinement and identifiability}

We next make the refinement step explicit.  The following elementary
construction is useful because it simultaneously gives nesting, boundaries of measure zero, and vanishing diameter.

\begin{lemma}[Nested continuity refinements of vanishing diameter]
\label{lem:nested_continuity_refinements}
Let $\mathcal C$ be a finite family of Borel cells in the compact metric
probability space $(\Theta,d_{\rm FS},\nu)$.  Assume that the cells partition
$\Theta$ up to a set of $\nu$-measure zero, that every cell has positive
$\nu$-measure, and that every cell boundary has $\nu$-measure zero.  For every
$\delta>0$, there is a finite refinement $\mathcal C'$ with the same
properties such that every cell of $\mathcal C'$ has diameter at most
$4\delta$.
\end{lemma}

\begin{proof}
Choose a finite $\delta$-net $x_1,\ldots,x_N$ in $\Theta$.  For each $r$, choose
a radius $t_r\in(\delta,2\delta)$ for which
$\nu(\partial B(x_r,t_r))=0$.  Such radii exist because the boundaries of balls
with distinct radii are disjoint, so only countably many radii can have
positive boundary measure.  The balls $B(x_r,t_r)$ still cover $\Theta$.
Disjointize this cover by setting
\[
 D_1=B(x_1,t_1),~
 D_r=B(x_r,t_r)\setminus\bigcup_{a<r}B(x_a,t_a),
~ 2\le r\le N.
\]
Every $D_r$ has boundary contained in the finite union of the chosen sphere
boundaries, and every nonempty $D_r$ has diameter at most $2t_r<4\delta$.
Intersecting the $D_r$ with the cells of $\mathcal C$ gives a finite family of
Borel sets whose boundaries have $\nu$-measure zero, because they are
contained in the union of the old cell boundaries and the chosen sphere
boundaries.  Retain only the intersections having positive $\nu$-measure.  The finite
union $N$ of the omitted intersections satisfies $\nu(N)=0$.  Consequently,
the retained cells cover $\Theta$ except for a set of $\nu$-measure zero.
Each retained cell lies in one $D_r$, so its diameter is less than
$4\delta$, and it lies in one parent cell, so the family is a refinement in
the measure-zero convention stated above.
\end{proof}

Apply Lemma~\ref{lem:nested_continuity_refinements} recursively to
$\mathcal C_1$, with $\delta_k\downarrow0$, and denote the resulting refinement
sequence by $(\mathcal C_k)_k$.  It is nested in the above measure-zero convention, every cell has positive
$\nu$-measure and boundary of $\nu$-measure zero, and
\begin{equation}
 \max_j\operatorname{diam}_{\rm FS}(C_{k,j})\longrightarrow0.
 \label{eq:appendix_mesh_diameter}
\end{equation}
No shape-regularity assumption is used.

\begin{lemma}[Spanning persists under refinement]
\label{lem:spanning_persists_refinement}
If the cell operators of $\mathcal C_k$ span the Hermitian matrices and
$\mathcal C_{k+1}$ refines $\mathcal C_k$, then the cell operators of
$\mathcal C_{k+1}$ also span the Hermitian matrices.
\end{lemma}

\begin{proof}
Because the refinement is understood up to sets of $\nu$-measure zero, each parent cell $C_{k,j}$ is the
disjoint union of its child cells of positive $\nu$-measure up to a set of $\nu$-measure zero.
Therefore
\[
 A_{k,j}=\sum_{\ell:\,C_{k+1,\ell}\subset C_{k,j}}A_{k+1,\ell}.
\]
Thus every parent-cell operator belongs to the linear span of the child-cell
operators.  The latter span therefore contains the span of the former.
\end{proof}

\begin{lemma}[Spanning implies global identifiability]
\label{lem:spanning_implies_identifiability}
If $\{A_{k,j}\}_j$ spans the Hermitian matrices, then the map
$\theta\mapsto q_{\theta,k}$ is globally identifiable.
\end{lemma}

\begin{proof}
Suppose $q_{\theta,k}(j)=q_{\theta',k}(j)$ for every $j$.  By
\eqref{eq:recovering_mesh_model},
\[
 \operatorname{Tr}\!\left[(\rho_\theta-\rho_{\theta'})A_{k,j}\right]=0
 \qquad\text{for every }j.
\]
Since the $A_{k,j}$ span the Hermitian matrices, the Hermitian operator
$\rho_\theta-\rho_{\theta'}$ is orthogonal in the Hilbert--Schmidt inner product
to every Hermitian matrix.  Hence it is zero.  Equality of the rank-one
projectors implies $\theta=\theta'$ as rays.
\end{proof}

Lemmas~\ref{lem:small_cell_spanning}--\ref{lem:spanning_implies_identifiability}
prove items 1--3 of Proposition~\ref{prop:recovering_meshes_exist}.

\subsection{Fine and mesh directional scores}

Fix $\theta\in\Theta$ and $w\in T_\theta\Theta$.  Let $t\mapsto\theta_t$ be a
smooth curve with $\theta_0=\theta$ and initial velocity $w$.  The fine
one-copy density with respect to $\nu$ is
\[
 p_\theta(v)=d\operatorname{Tr}(\rho_\theta\rho_v).
\]
Define its directional score by
\begin{equation}
 S_{\theta,w}(v)
 :=\left.\frac{d}{dt}\log p_{\theta_t}(v)\right|_{t=0},
 \label{eq:fine_directional_score}
\end{equation}
where its value on the zero set, which has $P_\theta$-measure zero of $p_\theta$ may be chosen
arbitrarily.  Although the logarithmic score may be singular on that zero
hyperplane, the product
\begin{equation}
 p_\theta(v)S_{\theta,w}(v)
 =\partial_w p_\theta(v)
 =d\operatorname{Tr}\!\left[(\partial_w\rho_\theta)\rho_v\right]
 \label{eq:score_times_density_derivative}
\end{equation}
is continuous and bounded in $v$ on the compact projective space.  More
precisely, for the fixed smooth curve $t\mapsto\theta_t$, smoothness and
compactness imply that $\partial_t p_{\theta_t}(v)$ is uniformly bounded in
$v$ for $t$ in a neighborhood of zero.  Dominated differentiation therefore
justifies differentiating each cell probability under the integral.  Moreover,
the known Fisher identity for the fine Haar observation is finite and gives
\begin{equation}
 \mathbb E_\theta[S_{\theta,w}(V)^2]
 =I_\theta^{\rm fine}[w,w]
 =2g_{{\rm FS},\theta}[w,w],
 \label{eq:fine_score_fisher_identity}
\end{equation}
so $S_{\theta,w}\in L^2(P_\theta)$.

Let $J_k$ be the index of the cell containing $V$.  The finite-alphabet
probabilities and directional scores are
\[
 q_{\theta,k}(j)=\int_{C_{k,j}}p_\theta(v)\,\nu(dv),
 ~
 S_{\theta,w}^{(k)}(j)
 :=\left.\frac{d}{dt}\log q_{\theta_t,k}(j)\right|_{t=0}.
\]
Every $q_{\theta,k}(j)$ is strictly positive.  Indeed, the zero set of
$v\mapsto p_\theta(v)$ is a projective hyperplane of invariant measure zero,
whereas every cell has positive measure.

\begin{lemma}[The mesh score is a conditional expectation]
\label{lem:mesh_score_conditional_expectation}
For every $k$, $\theta$, and $w$,
\begin{equation}
 S_{\theta,w}^{(k)}(J_k)
 =\mathbb E_\theta[S_{\theta,w}(V)\mid J_k]
 \quad P_\theta\text{-almost surely}.
 \label{eq:mesh_score_as_conditional_expectation}
\end{equation}
\end{lemma}

\begin{proof}
For a cell $C_{k,j}$, the bounded derivative in
\eqref{eq:score_times_density_derivative} justifies differentiation under the
cell integral and gives
\begin{align*}
 \left.\frac{d}{dt}q_{\theta_t,k}(j)\right|_{t=0}
& =\int_{C_{k,j}}\partial_w p_\theta(v)\,\nu(dv) \\
& =\int_{C_{k,j}}p_\theta(v)S_{\theta,w}(v)\,\nu(dv).
\end{align*}
Dividing by the positive number $q_{\theta,k}(j)$ yields
\[
 S_{\theta,w}^{(k)}(j)
 =\frac{\int_{C_{k,j}}p_\theta(v)S_{\theta,w}(v)\,\nu(dv)}
        {\int_{C_{k,j}}p_\theta(v)\,\nu(dv)},
\]
which is precisely the conditional mean of the fine score on that cell.
\end{proof}

\subsection{Pointwise Fisher recovery}
Let
\(
 \mathcal F_k:=\sigma(J_k).
\)
Since the partitions are nested, the sigma-fields
$(\mathcal F_k)_k$ form an increasing sequence.  We first identify the
sigma-field generated in the limit.  Let $x\ne y$.  Once the maximal cell
diameter is smaller than $d_{\rm FS}(x,y)$, no cell can contain both points.
Thus the partitions separate points.

More quantitatively, let $O\subset\Theta$ be open.  Outside the countable union of the residual sets of $\nu$-measure zero from the successive refinements, every
$x\in O$ belongs to a cell of every mesh.  Choose $r>0$ such that
$B(x,r)\subset O$.  For every sufficiently large $k$, the cell containing $x$
has diameter smaller than $r$ and is therefore contained in $O$.  Consequently,
\begin{align}
 O
 &=
 \bigcup\left\{
 C:\ C\in\bigcup_k\mathcal C_k,\ C\subset O
 \right\}\nonumber\\
 &\qquad\qquad\qquad\qquad\quad\text{up to a set of }\nu\text{-measure zero}.
\label{eq:open_sets_as_mesh_cells}
\end{align}
The family on the right is countable.  Since $P_\theta\ll\nu$, the same
identity holds up to a set of $P_\theta$-measure zero.  Hence
\begin{equation}
 \sigma\!\left(\bigcup_k\mathcal F_k\right)
 =
 \mathcal B(\Theta)
 \quad\text{after completion under }P_\theta.
 \label{eq:mesh_sigma_fields_generate_borel}
\end{equation}

By Lemma~\ref{lem:mesh_score_conditional_expectation}, the directional
score of the $k$-th mesh is the conditional expectation of the fine
directional score:
\begin{equation}
 S_{\theta,w}^{(k)}(J_k)
 =
 \mathbb E_\theta
 \left[
 S_{\theta,w}(V)\mid\mathcal F_k
 \right].
 \label{eq:mesh_score_conditional_expectation_sigma_field}
\end{equation}
Since $S_{\theta,w}\in L^2(P_\theta)$ and the limiting sigma-field in
\eqref{eq:mesh_sigma_fields_generate_borel} contains the fine observation
$V$, the $L^2$ martingale convergence theorem gives
\begin{equation}
 S_{\theta,w}^{(k)}(J_k)
 \longrightarrow
 S_{\theta,w}(V)
 \quad\text{in }L^2(P_\theta).
 \label{eq:mesh_score_L2_convergence}
\end{equation}

We now translate this score convergence into Fisher-information recovery.
Conditional expectation is the orthogonal projection in $L^2(P_\theta)$.
Therefore,
\[
 S_{\theta,w}(V)
 -
 S_{\theta,w}^{(k)}(J_k)
\]
is orthogonal to every $\mathcal F_k$-measurable square-integrable random
variable, and in particular to $S_{\theta,w}^{(k)}(J_k)$.  The Pythagorean
identity consequently gives
\begin{align}
 \mathbb E_\theta\!\left[S_{\theta,w}(V)^2\right]
 &=
 \mathbb E_\theta\!\left[
   \left(S_{\theta,w}^{(k)}(J_k)\right)^2
 \right]
 \notag\\
 &\quad+
 \mathbb E_\theta\!\left[
   \left(
     S_{\theta,w}(V)-S_{\theta,w}^{(k)}(J_k)
   \right)^2
 \right].
 \label{eq:fisher_score_pythagorean_identity}
\end{align}
By the definitions of the fine and mesh Fisher information, this is
equivalently
\begin{align}
 &I_\theta^{\rm fine}[w,w]-I_{\theta,k}[w,w]\nonumber\\
 &\qquad\qquad\qquad =
 \mathbb E_\theta\!\left[
   \left(
     S_{\theta,w}(V)-S_{\theta,w}^{(k)}(J_k)
   \right)^2
 \right].
 \label{eq:fisher_loss_as_score_projection_error}
\end{align}
The right-hand side tends to zero by
\eqref{eq:mesh_score_L2_convergence}.  Using the fine Fisher identity
\eqref{eq:fine_score_fisher_identity}, we therefore obtain
\begin{align}
 I_{\theta,k}[w,w]
 &=
 \mathbb E_\theta\!\left[
   \left(S_{\theta,w}^{(k)}(J_k)\right)^2
 \right]
 \longrightarrow
 I_\theta^{\rm fine}[w,w]\nonumber\\
 &=
 2g_{{\rm FS},\theta}[w,w].
 \label{eq:pointwise_fisher_recovery_appendix}
\end{align}

The conditional-expectation interpretation also makes monotonicity explicit.
Because conditional expectation is orthogonal projection in $L^2(P_\theta)$
and $\mathcal F_k\subset\mathcal F_{k+1}$,
\begin{equation}
 I_{\theta,k}[w,w]
 \le I_{\theta,k+1}[w,w]
 \le I_\theta^{\rm fine}[w,w].
 \label{eq:fisher_monotonicity_under_refinement}
\end{equation}
Thus refinement can only restore Fisher information lost by coarse-graining.

\subsection{Uniform recovery on the unit tangent bundle}

Let
\[
 UT\Theta
 :=\{(\theta,w):\theta\in\Theta,\ w\in T_\theta\Theta,
 g_{{\rm FS},\theta}[w,w]=1\}
\]
be the Fubini--Study unit tangent bundle.  Since $\Theta$ is compact and finite
dimensional, $UT\Theta$ is compact.  For fixed $k$, define
\[
 f_k(\theta,w):=I_{\theta,k}[w,w].
\]
For a fixed finite mesh, the Fisher quadratic form can be written without
logarithms as
\begin{equation}
 f_k(\theta,w)
 =\sum_{j=1}^{m_k}
 \frac{\bigl(\partial_w q_{\theta,k}(j)\bigr)^2}
      {q_{\theta,k}(j)}.
 \label{eq:finite_mesh_fisher_explicit}
\end{equation}
The alphabet is finite.  Each numerator and denominator depends continuously
on $(\theta,w)$, and every denominator is positive.  Hence $f_k$ is continuous
on $UT\Theta$.

By \eqref{eq:fisher_monotonicity_under_refinement}, the sequence $(f_k)_k$ is
pointwise nondecreasing.  By
\eqref{eq:pointwise_fisher_recovery_appendix}, it converges pointwise on
$UT\Theta$ to the continuous constant function $2$.  Dini's theorem therefore
gives
\begin{equation}
 \sup_{(\theta,w)\in UT\Theta}|f_k(\theta,w)-2|
 \longrightarrow0.
 \label{eq:dini_uniform_fisher_recovery}
\end{equation}
This is exactly the uniform Fisher-recovery condition
\eqref{eq:uniform_fisher_recovery}.  Together with the construction and
identifiability arguments above, it proves
Proposition~\ref{prop:recovering_meshes_exist}.

\section{Proof of the Trace-SWAP Identity}
\label{proof_fact:trace_swap_identity}

\begin{fact}[Trace-SWAP identity]
\label{fact:trace_swap_identity}
Let \(A\) be an operator on \(\mathcal H^{\otimes n}\), and let
\(S_{k,n+1}\) be the swap between the \(k\)-th tensor factor and an additional
\((n+1)\)-th factor.  Then
\begin{equation}
    \operatorname{Tr}_{1,\ldots,n+1}
    \left[
        (A\otimes I)S_{k,n+1}
    \right]
    =
    \operatorname{Tr}_{1,\ldots,n}[A].
    \label{eq:trace_swap_identity}
\end{equation}
\end{fact}

\begin{proof}
It is enough by linearity to verify the identity for simple tensor operators
\(A=A_1\otimes\cdots\otimes A_n\).  Since the swap exchanges the \(k\)-th and
\((n+1)\)-th tensor factors, the full trace factorizes as
\begin{align}
&\quad\operatorname{Tr}\left[(A\otimes I)S_{k,n+1}\right]\nonumber\\
    &\qquad\qquad\quad=
    \left(\prod_{i\ne k}\operatorname{Tr}A_i\right)
    \operatorname{Tr}_{k,n+1}\left[(A_k\otimes I)S_{k,n+1}\right].\nonumber
\end{align}
For the last two factors,
\(\operatorname{Tr}_{k,n+1}[(A_k\otimes I)S_{k,n+1}]=\operatorname{Tr}A_k\),
which follows immediately by expanding in an orthonormal basis.  Hence the
right-hand side is \(\prod_i\operatorname{Tr}A_i=\operatorname{Tr}A\).  Linearity
gives the claim for every operator \(A\) on \(\mathcal H^{\otimes n}\).
\end{proof}

\section{Proof of Fact \ref{fact:part_swap}}\label{proof_fact:part_swap}
We insert a complete orthonormal basis of tensor product states $|j_1 \dots j_n\rangle$ to execute the partial trace over the $n$ observed systems,
\begin{align*}
&\langle x | \Tr_{1:n} \left[ S_{k,n+1} (A \otimes I_{n+1})\right] | y \rangle \\
=& \sum_{j_1 \dots j_n} \langle j_1 \dots j_n, x | S_{k,n+1}  (A \otimes I_{n+1}) | j_1 \dots j_n, y \rangle.
\end{align*}

Acting with \(S_{k,n+1}\) on the bra exchanges the entries \(j_k\) and \(x\),
giving \(\langle j_1\dots x\dots j_n,j_k|\).

The last tensor factor then contributes \(\langle j_k|y\rangle\), while \(A\)
acts on the first \(n\) factors:
\begin{equation*}
= \sum_{j_1 \dots j_n} \langle j_1 \dots x \dots j_n | A | j_1 \dots j_k \dots j_n \rangle \langle j_k | y \rangle.
\end{equation*}

The sum over the specific basis index $j_k$ then performs the resolution of the identity on the $k$-th slot.  Equivalently, it replaces the ket entry in that slot by $y$, leaving
\[
    \sum_{j_{\setminus k}}
    \langle j_1 \dots x \dots j_n | A |
            j_1 \dots y \dots j_n \rangle,
\]
which is the claimed formula.  This completes the proof.\hfill\QED

\section{Proof of Lemma \ref{lem:m_k_perm}}\label{proof_lem:m_k_perm}
Apply Fact~\ref{fact:part_swap} with
\(A=\Pi_{\mathrm{sym}}^{(n)}O_n\) to the matrix element
\(\langle x|K_{k,\vec u}|y\rangle\):
\begin{equation*}
\langle x | K_{k,\vec{u}} | y \rangle = \sum_{j_{\setminus k}} \langle j_1 \dots x \dots j_n | \Pi_{\mathrm{sym}}^{(n)} O_n | j_1 \dots y \dots j_n \rangle.
\end{equation*}

Use
\(\Pi_{\mathrm{sym}}^{(n)}=n!^{-1}\sum_{\pi\in S_n}P_\pi\), with the convention
\begin{equation*}
    P_\pi|v_1\otimes\cdots\otimes v_n\rangle
    =
    |v_{\pi(1)}\otimes\cdots\otimes v_{\pi(n)}\rangle .
\end{equation*}
The omitted row is \(k\), and the target column is \(j=\pi(k)\).  Expanding
the two inner products gives
\begin{align*}
&\langle j_1 \dots x \dots j_n | u_{\pi(1)} \dots u_{\pi(n)} \rangle \nonumber\\
&\qquad\qquad\qquad\qquad= \left( \prod_{m \neq k} \langle j_m | u_{\pi(m)} \rangle \right) \langle x | u_{\pi(k)} \rangle,
\end{align*}
and,
\begin{equation*}
\langle u_1 \dots u_n | j_1 \dots y \dots j_n \rangle = \left( \prod_{m \neq k} \langle u_m | j_m \rangle \right) \langle u_k | y \rangle.
\end{equation*}

Substitution yields
\begin{align*}
&\langle x | K_{k,\vec{u}} | y \rangle\\
 =& \frac{1}{n!} \sum_{\pi \in S_n} \sum_{j_{\setminus k}} \left[ \left( \prod_{m \neq k} \langle u_m | j_m \rangle \langle j_m | u_{\pi(m)} \rangle \right) \langle x | u_{\pi(k)} \rangle \langle u_k | y \rangle \right].
\end{align*}

For each \(m\ne k\), the resolution of the identity gives
\begin{equation*}
\sum_{j_m} \langle u_m | j_m \rangle \langle j_m | u_{\pi(m)} \rangle = \langle u_m | I | u_{\pi(m)} \rangle = \langle u_m | u_{\pi(m)} \rangle.
\end{equation*}

Hence
\begin{align*}
&\hspace{-10pt}\langle x | K_{k,\vec{u}} | y \rangle\nonumber\\
&= \frac{1}{n!} \sum_{\pi \in S_n} \left( \prod_{m \neq k} \langle u_m | u_{\pi(m)} \rangle \right) \langle x | u_{\pi(k)} \rangle \langle u_k | y \rangle\\
&= \frac{1}{n!} \sum_{j=1}^n \sum_{\substack{\pi\in S_n:\\ \pi(k) = j}}\left( \prod_{m \neq k} \langle u_m | u_{\pi(m)} \rangle \right) \langle x | u_{\pi(k)} \rangle \langle u_k | y \rangle
\end{align*}

For fixed \(j=\pi(k)\), the restriction of \(\pi\) is a bijection from
\(\{1,\ldots,n\}\setminus\{k\}\) to
\(\{1,\ldots,n\}\setminus\{j\}\).  Therefore the sum of
\(\prod_{m\ne k}G_{m,\pi(m)}\) over these restrictions is
\(\Perm(G(k|j))\), and
\begin{equation*}
\langle x | K_{k,\vec{u}} | y \rangle = \frac{1}{n!} \sum_{j=1}^n \Perm(G(k|j)) \langle x | u_j \rangle \langle u_k | y \rangle.
\end{equation*}

Since this holds for all \(x,y\),
\begin{equation*}
K_{k,\vec{u}} = \frac{1}{n!} \sum_{j=1}^n \Perm(G(k|j)) |u_j\rangle\langle u_k|.
\end{equation*}
This completes the proof of Lemma \ref{lem:m_k_perm}.\hfill\QED

\section*{Acknowledgements} 
N.A.W. acknowledges the funding support from the National Quantum Mission, an initiative of the Department of Science and Technology, Govt. of India and the support provided by the Foundation for QC Innovation (FQCI), DST-NQM T-Hub at IISc Bengaluru, in facilitating this project. N.A.W also acknowledges the support provided by the Grant  ANRF/ARG/2025/012066/MS from the Department of Science \& Technology, Govt of India. 
A.D. was supported by the TCS Research Scholar Program (Cycle 19) from Tata Consultancy Services. 
M.H. was supported in part by the General R\&D Projects of 1+1+1
CUHK--CUHK(SZ)--GDST Joint Collaboration Fund (Grant No.~GRDP2025-022) and the
Guangdong Provincial Quantum Science Strategic Initiative (Grant
No.~GDZX2505003).
During the preparation of this manuscript, the authors used Microsoft Copilot with the GPT-5.5 and 5.6 Thinking model to assist with language editing, organization and presentation of the manuscript, and the exploration, development, and checking of certain mathematical derivations and arguments. 
All AI-assisted material was critically reviewed, verified, and revised by the authors, who take full responsibility for the accuracy and integrity of the manuscript.

\end{document}